\documentclass[11pt,a4paper]{article}
\usepackage{epstopdf}
\usepackage{graphicx,times,verbatim,amsmath,amssymb,amsthm}
\usepackage{natbib}
\usepackage{bm}
\usepackage{enumerate}
\usepackage{mathrsfs}
\usepackage{epsfig}
\usepackage{graphicx}
\usepackage{caption}
\usepackage{subcaption}
\usepackage{verbatim}
\usepackage{multirow}
\usepackage{array}
\usepackage{fontenc}
\usepackage{natbib}
\DeclareGraphicsExtensions{.eps}
\usepackage[top=1in, bottom=1in, left=0.7in, right=0.7in]{geometry}

\renewcommand\theequation{\arabic{section}.\arabic{equation}}

\newcommand{\vertiii}[1]{{\left\vert\kern-0.25ex\left\vert\kern-0.25ex\left\vert #1
    \right\vert\kern-0.25ex\right\vert\kern-0.25ex\right\vert}}
\newcommand{\mean}[1]{\mkern 1.5mu \overline{\mkern-2mu #1 \mkern-2mu}\mkern 1.5mu}
\newcommand{\thickhline}{%
    \noalign {\ifnum 0=`}\fi \hrule height 1pt
    \futurelet \reserved@a \@xhline
}

\newtheorem{theorem}{Theorem}[section]
\newtheorem{lemma}[theorem]{Lemma}
\newtheorem{corollary}[theorem]{Corollary}

\newtheorem{assumption}[theorem]{Assumption}

\newtheorem*{remark*}{Remark}

\DeclareMathOperator*{\argmax}{arg\,max}

\begin{document}

\title{\textbf{A Distributed One-Step Estimator}}
\author{\textbf{Cheng Huang} \thanks{c.huang@gatech.edu}\; and \textbf{Xiaoming Huo} \thanks{huo@gatech.edu} \\ \emph{Georgia Institute of Technology}}
\date{November 4, 2015}
\maketitle

\begin{abstract}
Distributed statistical inference has recently attracted enormous attention.
Many existing work focuses on the averaging estimator, e.g., \cite{zhang2013communication} together with many others.
We propose a one-step approach to enhance a simple-averaging based distributed estimator.
We derive the corresponding asymptotic properties of the newly proposed estimator.
We find that the proposed one-step estimator enjoys the same asymptotic properties as the centralized estimator.
The proposed one-step approach merely requires one additional round of communication in relative to the averaging estimator; so the extra communication burden is insignificant.
In finite sample cases, numerical examples show that the proposed estimator outperforms the simple averaging estimator with a large margin in terms of the mean squared errors.
A potential application of the one-step approach is that one can use multiple machines to speed up large scale statistical inference with little compromise in the quality of estimators.
The proposed method becomes more valuable when data can only be available at distributed machines with limited communication bandwidth.
\end{abstract}

\section{Introduction}
\label{intro}
In many important contemporary applications, data are often partitioned across multiple servers.
For example, a search engine company may have data coming from a large number of locations, and each location collects tera-bytes of data per day \citep{Corbett-et-al2012}.
On a different setting, high volumn of data (like videos) have to be stored distributively, instead of on a centralized server \citep{Mitra-et-al2011}.
Given the modern ``data deluge", it is often the case that centralized methods are no longer possible to implement.
It has also been notified by various researchers (e.g., \cite{jaggi2014communication}) that the speed of local processors can be thousands time faster than the rate of data transmission in a modern network.
Consequently it is evidently advantageous to develop communication-efficient method, instead of transmitting data to a central location and then apply a global estimator.

In statistical inference, estimators are introduced to infer some important hidden quantities.
In ultimate generality, a statistical estimator of a parameter $\theta \in \Theta$ is a measurable
function of the data, taking values in the parameter space $\Theta$.
Many statistical inference problems could be solved by finding the maximum likelihood estimators (MLE), or more generally, M-estimators.
In either case, the task is to minimize an objective function, which is the average of a criterion function over the entire data, which is typically denoted by $S = \{X_1, X_2, \ldots, X_N\}$, where $N$ is called the sample size.
Here we choose a capitalized $N$ to distinguish from a lower $n$ that will be used later.
Traditional centralized setting requires access to entire data set $S$ simultaneously.
However, due to the explosion of data size, it may be infeasible to store all the data in a single machine like we did during past several decades.
Distributed (sometimes, it is called {\it parallel}) statistical inference would be an indispensable approach for solving these large-scale problems.

At a high level, there are at least two types of distributed inference problems.
In the first type, each sample $X_i$ is completely observed at one location; at the same time, different samples (i.e., $X_i$ and $X_j$ for $i\neq j$) may be stored at different locations.
This paper will focus to this type of problems.
On the other hand, it is possible that for the same sample $X_i$, different parts are available at different locations, and they are {\it not} available in a centralized fashion.
The latter has been studied in the literature (see \cite{Lai2015} and references therein).
This paper will not study the second type.

For distributed inference in the first type of the aforementioned setting, data are split into several subsets and each subset is assigned to a processor.
This paper will focus on the M-estimator framework, in which an estimator is obtained by solving a distributed optimization problem.
The objective in the distributed optimization problem may come from an M-estimator framework (or more particularly from the maximum likelihood principle), empirical risk minimization, and/or penalized version of the above.
Due to the type 1 setting, we can see that the objective functions in the corresponding optimization problem are separable; in particular, the global objective function is a summation of functions such that each of them only depends on data reside on one machine.
The exploration in this paper will base on this fact.
As mentioned earlier, a distributed inference algorithm should be communication-efficient because of high communication cost between different machines or privacy concerns (such as sensitive personal information or financial data). It is worth noting that even if the data could be handled by a single machine, distributed inference would still be beneficial for reducing computing time.

Our work has been inspired by recent progress in distributed optimization. We review some noticeable progress in numerical approaches and their associated theoretical analysis.
Plenty of research work has been done in distributed algorithms for large scale optimization problems during recent years.
\cite{boyd2011distributed} suggests to use Alternating Direction Method of Multipliers (ADMM) to solve distributed optimization problems in statistics and machine learning.
Using a trick of {\it consistency} (or sometimes called {\it consensus}) constraints on local variables and a global variable, ADMM can be utilized to solve a distributed version of the Lasso problem \citep*{lasso,basispursuit}.
ADMM has also been adopted in solving distributed logistic regression problem, and many more.
ADMM is feasible for a wide range of problems, but it requires iterative communication between local machines and the center.
In comparison, we will propose a method that only requires two times iteration.
\cite{zinkevich2010parallelized} proposes a parallelized stochastic gradient descent method for empirical risk minimization and proves its convergence.
The established contractive mappings technique seems to be a powerful method to quantify the speed of convergence of the derived estimator to its limit.
\cite{shamir2013communication} presents the Distributed Approximate Newton-type Method (DANE) for distributed statistical optimization problems.
Their method firstly averages the local gradients then follows by averaging all local estimators in each iteration until convergence.
They prove that this method enjoys linear convergence rate for quadratic objectives.
For non-quadratic objectives, it has been showed that the value of objective function has geometric convergence rate.
\cite{jaggi2014communication} proposes a communication-efficient method for distributed optimization in machine learning, which uses local computation with randomized dual coordinate descent in a primal-dual setting.
They also prove the geometric convergence rate of their method.
The above works focused on the properties of numerical solutions to the corresponding optimization problems.
Nearly all of them require more than two rounds of communication.
Due to different emphasis, they did not study the statistical asymptotic properties (such as convergence in probability, asymptotic normality, Fisher information bound) of the resulting estimators.

Now we switch the gear to statistical inference. Distributed inference has been studied in many existing works, and various proposals have been made in different settings.
To the best of our knowledge, the distributed one-step estimator has {\it not} been studied in any of these existing works.
We review a couple of state-of-the-art approaches in the literature.
Our method builds on a closely related recent line of work of \cite{zhang2013communication}, which presents a straight forward approach to solve large scale statistical optimization problem, where the local empirical risk minimizers are simply averaged.
They showed that this averaged estimator achieves mean squared error that decays as $O(N^{-1} + (N/k)^{-2})$, where $N$ stands for the total number of samples and $k$ stands for the total number of machines.
They also showed that the mean squared error could be even reduced to $O(N^{-1} + (N/k)^{-3})$ with one more bootstrapping sub-sampling step.
Obviously, there exists efficiency loss in their method since the centralized estimator could achieve means squared error $O(N^{-1})$.
%However, to the best of our knowledge, current literature does not provide a distributed approach, which could achieve the same mean squared error as centralized approaches.
\cite{liu2014distributed} proposes an inspiring two-step approach: firstly find local maximum likelihood estimators, then subsequently combine them by minimizing the total Kullback-Leibler divergence (KL-divergence).
They proved the exactness of their estimator as the global MLE for the full exponential family.
They also estimated the mean squared errors of the proposed estimator for a curved exponential family.
Due to the adoption of the KL-divergence, the effectiveness of this approach heavily depends on the parametric form of the underlying model.
\cite{chen2014split} proposes a split-and-conquer approach for a penalized regression problem (in particular, a model with the canonical exponential distribution) and show that it enjoys the same oracle property as the method that uses the entire data set in a single machine.
Their approach is based on a majority voting, followed by a weighted average of local estimators, which somewhat resembles a one-step estimator however is different.
In addition, their theoretical results requires $k \le O(N^{\frac{1}{5}})$, where $k$ is the number of machines and $N$ is the total number of samples; this is going to be different from our needed condition for theoretical guarantees.
Their work considers a high-dimensional however sparse parameter vector, which is not considered in this paper.
\cite{rosenblatt2014optimality} analyzes the error of averaging estimator in distributed statistical learning under two scenarios.
The number of machines is fixed in the first one and the number of machines grows in the same order with the number of samples per machine.
They presented asymptotically exact expression for estimator error in both scenarios and showed that the error grows linearly with the number of machines in the latter case.
Their work does not consider the one-step updating that will be studied in this paper.
Although it seems that their work proves the asymptotic optimality of the simple averaging, our simulations will demonstrate the additional one-step updating can improve over the simple averaging, at least in some interesting finite sample cases.
\cite{battey2015distributed} study the distributed parameter estimation method for penalized regression and establish the oracle asymptotic property of an averaging estimator.
They also discussed hypotheses testing, which is not covered in this paper.
Precise upper bounds on the errors of their proposed estimator have been developed.
We benefited from reading the technical proofs of their paper; however unlike our method, their method is restricted to linear regression problems with penalty and requires the number of machine $k=o(\sqrt{N})$.
\cite{lee2015communication} devise a one-shot approach, which averages ``debiased"  lasso estimators, to distributed sparse regression in the high-dimensional setting.
They show that their approach converges at the same order of rate as the Lasso when the data set is not split across too many machines.
%\cite{van2014asymptotically}

It is worth noting that near all existing distributed estimator are {\it averaging} estimators.
The idea of applying one additional updating, which correspondingly requires one additional round of communication, has not be explicitly proposed.
We may notice some precursor of this strategy.
For example, in \cite{shamir2013communication}, an approximate Newton direction was estimated at the central location, and then broadcasted to local machines.
Another occurrence is that in \cite{lee2015communication}, some intermediate quantities are estimated in a centralized fashion, and then distributed to local machines.
None of them explicitly described what we will propose.

In the theory on maximum likelihood estimators (MLE) and M-estimators, there is a one-step method, which could make a consistent estimator as efficient as MLE or M-estimators with a single Newton-Raphson iteration.
(Here, efficiency stands for the relative efficiency converges to $1$.)
See \cite{van2000asymptotic} for more details.
There have been numerous papers utilizing this method.
See \cite{bickel1975one}, \cite{fan1999one} and \cite{zou2008one}.
One-step estimator enjoys the same asymptotic properties as the MLE or M-estimators as long as the initial estimators are $\sqrt{n}$-consistent.
A $\sqrt{n}$-consistent estimator is much easier to find than the MLE or an M-estimator.
For instance, the simple averaging estimator (e.g., the one proposed by \cite{zhang2013communication}) is good enough as a starting point for a one-step estimator.

In this paper, we propose a one-step estimator for distributed statistical inference.
The proposed estimator is built on the well-analyzed simple averaging estimator.
We show that the proposed one-step estimator enjoys the same asymptotic properties (including convergence and asymptotic normality) as the centralized estimator, which would utilize the entire data.
Given the amount of knowledge we had on the distributed estimators, the above result may not be surprising.
However, when we derive an upper bound for the error of the proposed one-step estimator, we found that we can achieve a slightly better one than those in the existing literature.
We also perform a detailed evaluation of our one-step method, comparing with simple averaging method and centralized method using synthetic data.
The numerical experiment is much more encouraging than the theory predicts: in nearly all cases, the one-step estimator outperformed the simple averaging one with a clear margin.
We also observe that the one-step estimator achieves the comparable performance as the global estimator at a much faster rate than the simple averaging estimator.
Our work may indicate that in practice, it is better to apply a one-step distributed estimator, than a simple-average one.

This paper is organized as follows.
Section \ref{background} describes details of our problem setting and two methods---the simple averaging method and the proposed one-step method.
In Section \ref{asym}, we study the asymptotic properties of the one-step estimator in the M-estimator framework and analyze the upper bound of its estimation error.
Section \ref{numerical} provides some numerical examples of distributed statistical inference with synthetic data.
We conclude in Section \ref{conclusion}.
When appropriate, detailed proofs are relegated to the appendix.

\section{Problem Formulation}
\label{background}

\subsection{Notations}
In this subsection, we will introduce some notations that will be used in this paper.
Let $\{ m(x;\theta): \theta \in \Theta \subset \mathbb{R}^d \}$ denote a collection of criterion functions, which should have continuous second derivative. Consider a data set $S$ consisting of $N=nk$ samples, which are drawn i.i.d. from $p(x)$ (for simplicity, we assume that the sample size $N$ is a multiple of $k$). This data set is divided evenly at random and stored in $k$ machines. Let $S_i$ denote the subset of data assigned to machine $i$, $i=1,\ldots,k$, which is a collection of $n$ samples drawn i.i.d. from $p(x)$. Note that any two subsets in those $S_i$'s are not overlapping.\\
For each $i \in \{1,\ldots,k\}$, let the local empirical criterion function that is based on the local data set on machine $i$ and the corresponding maximizer be denoted by
\begin{equation}
M_i(\theta) = \frac{1}{|S_i|} \sum_{x \in S_i} m(x;\theta) \text{\; and \;} \theta_i = \argmax_{\theta \in \Theta} M_i(\theta).
\label{def_local}
\end{equation}
Let the global empirical criterion function be denoted by
\begin{equation}
M(\theta) = \frac{1}{k} \sum_{i=1}^k M_i(\theta).
\end{equation}
And let the population criterion function and its maximizer be denoted by
\begin{equation}
M_{0}(\theta) = \int_{\mathcal{X}} m(x;\theta)p(x) dx \text{\; and \;} \theta_0 = \argmax_{\theta \in \Theta} M_0(\theta),
\label{def_pop}
\end{equation}
where $\mathcal{X}$ is the sample space. Note that $\theta_0$ is the parameter of interest. The gradient and Hessian matrix of $m(x;\theta)$ with respect to $\theta$ are denoted by
\begin{equation}
\dot{m}(x;\theta) = \frac{\partial m(x;\theta)}{\partial \theta},\ddot{m}(x;\theta) = \frac{\partial^2 m(x;\theta)}{\partial \theta \; \partial \theta^T}.
\label{def_gh}
\end{equation}
We also let the gradient and Hessian of local empirical criterion function be denoted by
\begin{equation}
\dot{M}_i(\theta) = \frac{\partial M_i(\theta)}{\partial \theta} = \frac{1}{|S_i|} \sum_{x \in S_i}\frac{\partial m(x;\theta)}{\partial \theta},
\ddot{M}_i(\theta) = \frac{\partial^2 M_i(x;\theta)}{\partial \theta \; \partial \theta^T} = \frac{1}{|S_i|} \sum_{x \in S_i} \frac{\partial^2 m(x;\theta)}{\partial \theta \; \partial \theta^T},
\label{def_gh_local}
\end{equation}
where $i \in \{1,2,\ldots,k\}$, and let the gradient and Hessian of global empirical criterion function be denoted by
\begin{equation}
\dot{M}(\theta) = \frac{\partial M(\theta)}{\partial \theta}, \ddot{M}(\theta) = \frac{\partial^2 M(\theta)}{\partial \theta \; \partial \theta^T}.
\label{def_gh_global}
\end{equation}
Similarly, let the gradient and Hessian of population criterion function be denoted by
\begin{equation}
\dot{M}_0(\theta) = \frac{\partial M_0(\theta)}{\partial \theta}, \ddot{M}_0(\theta) = \frac{\partial^2 M_0(\theta)}{\partial \theta \; \partial \theta^T}.
\label{def_gh_pop}
\end{equation}

The vector norm $\| \cdot \|$ for $a \in \mathbb{R}^d$ that we use in this paper is the usual Euclidean norm $\| a \| = (\sum_{j=1}^d a_j^2)^{\frac{1}{2}}$. And we also use $\vertiii{\cdot}$ to denote a norm for matrix $A \in \mathbb{R}^{d \times d}$, which is defined as its maximal singular value, i.e., we have
$$\vertiii{A} = \sup_{u: u \in R^d, \|u\| \le 1} \| Au \| .$$
The aforementioned matrix norm will be the major matrix norm that is used throughout the paper. The only exception is that we will also use Frobenius norm in Appendix \ref{bounds}. And the Euclidean norm is the only vector norm that we use throughout this paper.

\subsection{Review on M-estimators}
In this paper, we will study the distributed scheme for large-scale statistical inference. To make our conclusions more general, we consider M-estimators, which could be regarded as a generalization of the Maximum Likelihood Estimators (MLE). 
The M-estimator $\hat{\theta}$ could be obtained by maximizing empirical criterion function, which means
$$
\hat{\theta} = \argmax_{\theta \in \Theta} M(\theta) = \argmax_{\theta \in \Theta} \frac{1}{|S|}\sum_{x \in S} m(x;\theta).
$$
Note that, when the criterion function is the log likelihood function, i.e., $m(x;\theta) = \log f(x;\theta)$, the M-estimator is exactly the MLE.
Let us recall that $M_{0}(\theta) = \int_{\mathcal{X}} m(x;\theta)p(x) dx$ is the population criterion function and $\theta_0= \argmax_{\theta \in \Theta} M_{0}(\theta)$ is the maximizer of population criterion function. It is known that $\hat{\theta}$ is a consistent estimator for $\theta_0$, i.e., $\hat{\theta} - \theta_0 \xrightarrow{P} 0$. See Chapter 5 of \cite{van2000asymptotic}.

\subsection{Simple Averaging Estimator}
Let us recall that $M_i(\theta)$ is the local empirical criterion function on machine $i$,
\begin{equation*}
M_i(\theta) = \frac{1}{|S_i|} \sum_{x \in S_i} m(x;\theta).
\end{equation*}
And, $\theta_i$ is the local M-estimator on machine $i$,
\begin{equation*}
\theta_i = \argmax_{\theta \in \Theta} M_i(\theta).
\end{equation*}
Then as mentioned in \cite{zhang2013communication}, the simplest and most intuitive method is to take average of all local M-estimators. Let $\theta^{(0)}$ denote the average of these local M-estimators, we have
\begin{equation}
\theta^{(0)} = \frac{1}{k} \sum_{i=1}^k \theta_i \label{def_theta(0)},
\end{equation}
which is referred as the simple averaging estimator in the rest of this paper. 

\subsection{One-step Estimator}
Under the problem setting above, starting from the simple averaging estimator $\theta^{(0)}$, we can obtain the one-step estimator $\theta^{(1)}$ by performing a single Newton-Raphson update, i.e.,
\begin{equation}
\theta^{(1)} = \theta^{(0)} - [\ddot{M}(\theta^{(0)})]^{-1}[\dot{M}(\theta^{(0)})], \label{def_theta(1)}
\end{equation}
where $M(\theta) = \frac{1}{k}\sum_{i=1}^k M_i(\theta)$ is the global empirical criterion function, $\dot{M}(\theta)$ and $\ddot{M}(\theta)$ are the gradient and Hessian of $M(\theta)$, respectively. The whole process to compute one-step estimator can be summarized as follows. 
\begin{enumerate}
	\item For each $i \in \{ 1,2,\ldots,k \}$, machine $i$ compute the local M-estimator with its local data set, 
$$\theta_i = \argmax_{\theta \in \Theta} M_i(\theta) = \argmax_{\theta \in \Theta} \frac{1}{|S_i|} \sum_{x \in S_i} m(x;\theta).$$
	\item All local M-estimators are averaged to obtain simple averaging estimator,
$$\theta^{(0)} = \frac{1}{k} \sum_{i=1}^k \theta_i \, .$$
Then $\theta^{(0)}$ is sent back to each local machine.
	\item For each $i \in \{ 1,2,\ldots,k \}$, machine $i$ compute the gradient and Hessian matrix of its local empirical criterion function $M_i(\theta)$ at $\theta = \theta^{(0)}$. Then send $\dot{M}_i(\theta^{(0)})$ and $\ddot{M}_i(\theta^{(0)})$ to the central machine.
	\item Upon receiving all gradients and Hessian matrices, the central machine computes gradient and Hessian matrix of $M(\theta)$ by averaging all local information,
$$
\dot{M}(\theta^{(0)}) = \frac{1}{k} \sum_{i=1}^k \dot{M}_i(\theta^{(0)}),\; \ddot{M}(\theta^{(0)}) = \frac{1}{k} \sum_{i=1}^k \ddot{M}_i(\theta^{(0)}).
$$
Then the central machine would perform a Newton-Raphson iteration to obtain a one-step estimator,
$$\theta^{(1)} = \theta^{(0)} - [\ddot{M}(\theta^{(0)})]^{-1}[\dot{M}(\theta^{(0)})] .$$
\end{enumerate}
Note that $\theta^{(1)}$ is not necessarily the maximizer of empirical criterion function $M(\theta)$ but it shares the same asymptotic properties with the corresponding global maximizer (M-estimator) under some mild conditions, i.e., we will show
$$
\theta^{(1)} \xrightarrow{P} \theta_0, \sqrt{N}(\theta^{(1)} - \theta_0) \xrightarrow{d} \bold{N}(0,\Sigma), \text{ as } N \rightarrow \infty, 
$$
where the covariance matrix $\Sigma$ will be specified later. \\
The one-step estimator has advantage over simple averaging estimator in terms of estimation error. In \cite{zhang2013communication}, it is showed both theoretically and empirically that the MSE of simple averaging estimator $\theta^{(0)}$ grows significantly with the number of machines $k$ when the total number of samples $N$ is fixed. More precisely, there exists some constant $C_1 , C_2 > 0$ such that
$$
\mathbb{E} [ \| \theta^{(0)} - \theta_0 \|^2 ] \le \frac{C_1}{N} + \frac{C_2 k^2}{N^2} + O(k N^{-2}) +O(k^3 N^{-3}).
$$
Fortunately, one-step method $\theta^{(1)}$ could achieve a lower upper bound of MSE with only one additional step. we will show the following in Section \ref{asym}:
$$
\mathbb{E} [ \| \theta^{(1)} - \theta_0 \|^2 ] \le \frac{C_1}{N} + O(N^{-2}) + O(k^4 N^{-4}).
$$

\section{Main Results}
\label{asym}
At first, some assumptions will be introduced in Section 3.1. After that, we will study the asymptotic properties of one-step estimator in Section 3.2, i.e., convergence,  asymptotic normality and mean squared error (MSE). In Section 3.3, we will consider the one-step estimator under the presence of information loss.

\subsection{Assumptions}
Throughout this paper, we impose some regularity conditions on the criterion function $m(x;\theta)$, the local empirical criterion function $M_i(\theta)$ and population criterion function $M_0(\theta)$. We use the similar assumptions in \cite{zhang2013communication}. Those conditions are also standard in classical statistical analysis of M-estimators (cf. \cite{van2000asymptotic}).

First assumption restricts the parameter space to be compact, which is reasonable and not rigid in practice. One reason is that the possible parameters lie in a finite scope for most cases. Another justification is that the largest number that computers could cope with is always limited.
\begin{assumption}[parameter space]
\label{parameter_space}
The parameter space $\Theta \in \mathbb{R}^d$ is a compact convex set. And let $D \triangleq \max_{\theta,\theta' \in \Theta} \| \theta - \theta' \|$ denote the diameter of $\Theta$.
\end{assumption}

We also assume that $m(x;\theta)$ is concave with respect to $\theta$ and $M_0(\theta)$ has some curvature around the unique optimal point $\theta_0$, which is a standard assumption for any method requires consistency.
\begin{assumption}[invertibility]
\label{invertibility}
The Hessian of population criterion function $M_0(\theta)$ at $\theta_0$ is a nonsingular matrix, which means $\ddot{M}(\theta_0)$ is negative definite and there exists some $\lambda>0$ such that $\sup_{u \in \mathbb{R}^d : \|u\|<1} u^t \ddot{M}(\theta_0) u \le -\lambda$.
\end{assumption}

In addition, we require the criterion function $m(x;\theta)$ to be smooth enough, at least in the neighborhood of the optimal point $\theta_0$, $B_\delta = \{ \theta \in \Theta: \| \theta - \theta_0 \| \le \delta \}$. So, we impose some regularity conditions on the first and second derivative of $m(x;\theta)$. We assume the gradient of $m(x;\theta)$ is bounded in moment and the difference between $\ddot{m}(x;\theta)$ and $\ddot{M}_0(\theta)$ is also bounded in moment. Moreover, we assume that $\ddot{m}(x;\theta)$ has Lipschitz continuity in $B_\delta$.
\begin{assumption}[smoothness]
\label{smoothness}
There exist some constants $G$ and $H$ such that
$$
\mathbb{E}[\| \dot{m}(X;\theta) \|^8] \le G^8 \text{\; and \;} \mathbb{E} \left[ \vertiii{ \ddot{m}(X;\theta) - \ddot{M}_0(\theta)}^8 \right] \le H^8, \forall \theta \in B_\delta.
$$
For any $x \in \mathcal{X}$, the Hessian matrix $\ddot{m}(x;\theta)$ is $L(x)$-Lipschitz continuous,
$$
\vertiii{ \ddot{m}(x;\theta) - \ddot{m}(x;\theta')} \le L(x) \| \theta - \theta' \|, \forall \theta,\theta' \in B_\delta,
$$
where $L(x)$ satisfies
$$
\mathbb{E}[L(X)^8] \le L^8 \text{\; and \;} \mathbb{E}[(L(X)-\mathbb{E}[L(X)])^8] \le L^8,
$$
for some finite constant $L>0$.
\end{assumption}

By Theorem 8.1 in Chapter XIII of \cite{lang1993real}, $m(x;\theta)$ enjoys interchangeability between differentiation on $\theta$ and integration on $x$, which means the following two equations hold:
$$
\dot{M}_0(\theta) = \frac{\partial}{\partial \theta} \int_{\mathcal{X}} m(x;\theta)p(x) dx = \int_\mathcal{X} \frac{\partial m(x;\theta)}{\partial \theta} p(x) dx = \int_\mathcal{X} \dot{m}(x;\theta)p(x) dx,
$$
and,
$$
\ddot{M}_0(\theta) = \frac{\partial^2}{\partial \theta^t \partial \theta} \int_{\mathcal{X}} m(x;\theta)p(x) dx = \int_\mathcal{X} \frac{\partial^2 m(x;\theta)}{\partial \theta^t \partial \theta} p(x) dx = \int_\mathcal{X} \ddot{m}(x;\theta)p(x) dx.
$$

\subsection{Asymptotic Properties and Mean Squared Error (MSE) Bound}
Our main result is that one-step estimator enjoys oracle asymptotic properties and has mean squared error of $O(N^{-1})$ under some mild conditions.
\begin{theorem}
\label{one_step_M_ap}
Let $\Sigma= \ddot{M}_0(\theta_0)^{-1} \mathbb{E}[\dot{m}(x;\theta_0)\dot{m}(x;\theta_0)^t] \ddot{M}_0(\theta_0)^{-1}$, where the expectation is taken with respect to $p(x)$. Under Assumption \ref{parameter_space}, \ref{invertibility}, and \ref{smoothness}, when the number of machines $k$ satisfies $k =O(\sqrt{N})$, $\theta^{(1)}$ is consistent and asymptotically normal, i.e., we have
$$
\theta^{(1)} - \theta_0 \xrightarrow{P} 0 \text{\; and \;}
\sqrt{N} (\theta^{(1)} - \theta_0) \xrightarrow{d} \bold{N}(0,\Sigma) \text{ as } N \rightarrow \infty.
$$
\end{theorem}
See Appendix \ref{proof_ap} for a proof.
The above theorem indicates that the one-step estimator is asymptotically equivalent to the centralized M-estimator.
\begin{remark*}
It is worth noting that the condition $\| \sqrt{N} (\theta^{(0)} - \theta_0) \|=O_P(1)$ suffices for our proof to Theorem \ref{one_step_M_ap}. Let $\tilde{\theta}^{(0)}$ denote another starting point for the one-step update, then the following estimator
$$
\tilde{\theta}^{(1)} = \tilde{\theta}^{(0)} - \ddot{M}(\tilde{\theta}^{(0)})^{-1} \dot{M}(\tilde{\theta}^{(0)})
$$
also enjoys the same asymptotic properties with $\theta^{(1)}$ (and the centralized M-estimator $\hat{\theta}$) as long as $\sqrt{N} (\tilde{\theta}^{(0)} - \theta_0)$ is bounded in probability. Therefore, we can replace $\theta^{(0)}$ with any estimator $\tilde{\theta}^{(0)}$ that satisfies
$$
\| \sqrt{N} (\tilde{\theta}^{(0)} - \theta_0) \| = O_P(1).
$$
\end{remark*}

\begin{theorem}
\label{one_step_M}
Under Assumption \ref{parameter_space}, \ref{invertibility}, and \ref{smoothness}, the mean squared error of the one-step estimator $\theta^{(1)}$ is bounded by
$$
\mathbb{E}[ \| \theta^{(1)} - \theta_0 \|^2 ] \le \frac{2 \mbox{Tr} [\Sigma]}{N} + O(N^{-2}) + O(k^{4}N^{-4}).
$$
When the number of machines $k$ satisfies $k =O(\sqrt{N})$, we have
$$
\mathbb{E}[ \| \theta^{(1)} - \theta_0 \|^2 ] \le \frac{2 \mbox{Tr} [\Sigma]}{N} + O(N^{-2}).
$$
\end{theorem}
See Appendix \ref{proof} for a proof. \\

\begin{comment}
\begin{remark*}
In \cite{battey2015distributed}, Theorem 4.10 claims that
$$
\| \bar{\beta} - \hat{\beta} \| =O_P(\sqrt{k}(d \vee \log n)/n), \| \bar{\beta} - \beta^{\star} \| =O_P(d/n).
$$
But it is restricted to linear models $Y_i = \bold{X}_i^t \beta^{\star} + \epsilon_i$ and requires the following conditions:
\begin{itemize}
	\item $\{(Y_i, X_i)\}^n_{i=1}$ are i.i.d. and $\Sigma$ satisfies $0 < C_{min} \le \lambda_{min}(\Sigma) \le \lambda_{max}(\Sigma) \le C_{max}$.
	\item The rows of $X$ are sub-Gaussian with $\|X_i\|_{\psi_2} \le \kappa, i = 1,\ldots,n$.
	\item $\{\epsilon_i\}^n_{i=1}$ are i.i.d sub-Gaussian random variables with $\| \epsilon_i \|_{\psi_2} \le \sigma_1$.
	\item number of subsamples satisfies $k = O(nd/(d \vee \log n)^2)$.
\end{itemize}
Note that sub-Gaussian means that every moment of the random variables/vectors exits. All the conditions except the last one are more strict than our assumptions. But the result is weaker than ours because convergence in moment does imply convergence in probability, but not vice versa. The proof techniques look very similar to ours.
\end{remark*}
\end{comment}

In particular, when we choose the criterion function to be the log likelihood function, $m(x;\theta) = \log f(x;\theta)$, the one-step estimator has the same asymptotic properties with the maximum likelihood estimator (MLE), which is described below.
\begin{corollary}
\label{one_step_mle}
If $m(x;\theta) = \log f(x;\theta)$ and $k = O(\sqrt{N})$, one-step estimator $\theta^{(1)}$ is a consistent and asymptotic efficient estimator of $\theta_0$,
$$
\theta^{(1)} - \theta_0 \xrightarrow{P} 0 \text{\; and \;} 
\sqrt{N} (\theta^{(1)} - \theta_0) \xrightarrow{d} \bold{N}(0,I(\theta_0)^{-1}), \text{ as } N \rightarrow \infty,
$$
where $I(\theta_0)$ is the Fisher's information at $\theta = \theta_0$. And the mean squared error of $\theta^{(1)}$ is bounded as follows:
$$
\mathbb{E}[ \| \theta^{(1)} - \theta_0 \|^2 ] \le \frac{2 \mbox{Tr} [I^{-1}(\theta_0)]}{N} + O(N^{-2}) + O(k^4 N^{-4}).
$$
\end{corollary}
\begin{proof}
It follows immediately from Theorem \ref{one_step_M_ap}, \ref{one_step_M} and the definition of the Fisher's information.
\end{proof}

\subsection{Under the Presence of Communication Failure}
In practice, it is possible that the information (local estimator, local gradient and local Hessian) from a local machine {\it cannot} be received by the central machine due to various causes (for instance, network problem or hardware crash). We assume that the communication failure on each local machine occurs independently.

We now derive a distributed estimator under the scenario with possible information loss. We will also present the corresponding theoretical results. We use $a_i \in \{0,1\}, i=1,\ldots,k$, to denote the status of local machines: when machine $i$ successfully sends all its local information to central machine, we have $a_i=1$; when machine $i$ fails, we have $a_i=0$. The corresponding simple averaging estimator is computed as
$$
\theta^{(0)} = \frac{\sum_{i=1}^k a_i \theta_i }{\sum_{i=1}^k a_i}.
$$
And one-step estimator is as follows
$$
\theta^{(1)} = \theta^{(0)} - \left[ \sum_{i=1}^k a_i \ddot{M}_i(\theta^{(0)}) \right]^{-1} \left[ \sum_{i=1}^k a_i \dot{M}_i(\theta^{(0)}) \right].
$$
\begin{corollary}
\label{estimation_with_crash}
Suppose $r$ is the probability (or rate) that a local machine fails to send its information to the central machine. When $n=N/k \rightarrow \infty$, $k \rightarrow \infty$ and $k=O(\sqrt{N})$, the one-step estimator is asymptotically normal:
\begin{eqnarray*}
\sqrt{(1-r)N}(\theta^{(1)} - \theta_0) \xrightarrow{d} \bold{N}(0,\Sigma).
\end{eqnarray*}
And more precisely, unless all machines fail, we have
$$
\mathbb{E}[ \| \theta^{(1)} - \theta_0 \|^2 ] \le \frac{2 \mbox{Tr} [\Sigma]}{N(1-r)} + \frac{6 \mbox{Tr} [\Sigma]}{N k (1-r)^2} + O(N^{-2}(1-r)^{-2}) + O(k^2 N^{-2}).
$$
\end{corollary}
See Appendix \ref{proof_crash} for a proof. Note that the probability that all machines fail is $r^k$, which is negligible when $r$ is small and $k$ is large.

\section{Numerical Examples}
\label{numerical}
In this section, we will discuss the results of simulation studies comparing the performance of the simple averaging estimator $\theta^{(0)}$ and the one-step estimator $\theta^{(1)}$, as well as the centralized M-estimator $\hat{\theta}$, which maximizes the global empirical criterion function $M(\theta)$ when the entire data are available centrally. Besides, we will also study the resampled averaging estimator, which is proposed by \cite{zhang2013communication}. The main idea of a resampled averaging estimator is to resample $\lfloor sn \rfloor$ observations from each local machine to obtain another averaging estimator $\theta^{(0)}_{1}$. Then the resampled averaging estimator can be constructed as follows:
$$
\theta^{(0)}_{re} = \frac{\theta^{(0)} - s\theta^{(0)}_1}{1-s}.
$$

In our numerical examples, the resampling ratio $s$ is chosen to be $s=0.1$ based on past empirical studies.
We shall implement these estimators for logistic regression, Beta distribution and Gaussian Distribution. We will also study the parameter estimation for Beta distribution with occurrence of communication failures, in which some local machines could fail to send their local information to the central machine.

\subsection{Logistic Regression}
In this example, we simulate the data from the following logistic regression model:
\begin{equation}
\label{logit_model}
y \sim \mbox{Bernoulli}(p), \text{ where } p = \frac{\exp(x^t \theta)}{1+\exp(x^t \theta)} = \frac{\exp(\sum_{j=1}^d x_j \theta_j)}{1+\exp(\sum_{j=1}^d x_j \theta_j)}.
\end{equation}
In this model, $y \in \{0,1\}$ is a binary response, $x \in R^d$ is a continuous predictor and $\theta \in R^d$ is the parameter of interest.

In each single experiment, we choose a fixed vector $\theta$ with each entry $\theta_j, j=1,\ldots,d,$ drawn from $\mbox{Unif} (-1,1)$ independently. Entry $x_j, j=1,\ldots,d$ of $x \in R^d$ is sampled from $\mbox{Unif}(-1,1)$, independent from parameters $\theta_j$'s and other entries. After generating parameter $\theta$ and predictor $x$, we can compute the value of probability $p$ and generate $y$ according to (\ref{logit_model}). We fix the number of observed samples $N=2^{17}=131,072$ in each experiment, but vary the number of machines $k$. The target is to estimate $\theta$ with different number of parallel splits $k$ of the data. The experiment is repeated for $K=50$ times to obtain reliable average error. And the criterion function is the log-likelihood function,
$$
m(x,y;\theta) = y x^t \theta - \log(1+\exp(x^t \theta)).
$$

The goal of each experiment is to estimate parameter $\theta_0$ maximizing population criterion function 
$$M_0(\theta) = \mathbb{E}_{x,y}[m(x,y;\theta)] = \mathbb{E}_{x,y}[y x^t \theta - \log(1+\exp(x^t \theta))].$$
In this particular case, $\theta_0$ is exactly the same with the true parameter.

In each experiment, we split the data into $k=2,4,8,16,32,64,128$ non-overlapping subsets of size $n=N/k$. We compute a local estimator $\theta_i$ from each subset. And simple averaging estimator is obtained by taking all local estimators, $\theta^{(0)} = \frac{1}{k}\sum_{i=1}^k \theta_i$. Then the one-step estimator $\theta^{(1)}$ could be computed by applying a Newton-Raphson update to $\theta^{(0)}$, i.e., equation (\ref{def_theta(1)}).

The dimension is chosen to be $d=20$ and $d=100$, which could help us understand the performance of those estimators in both low and high dimensional cases. In Fig. \ref{logistic}, we plot the mean squared error of each estimator versus the number of machines $k$.  As we expect, the mean squared error of simple averaging estimator grows rapidly with the number of machines. But, the mean squared error of one-step estimator remains the same with the mean squared error of oracle estimator when the number of machines $k$ is not very large. Even when the $k=128$ and the dimension of predictors $d=100$, the performance of one-step estimator is significantly better than simple averaging estimator. As we can easily find out from Fig. \ref{logistic}, the mean squared error of simple averaging estimator is about $10$ times of that of one-step estimator when $k=128$ and $d=100$.
\begin{figure}[h!]
	\centering
	\begin{subfigure}[b]{0.6\textwidth}
		\centering
		\includegraphics[width=\textwidth]{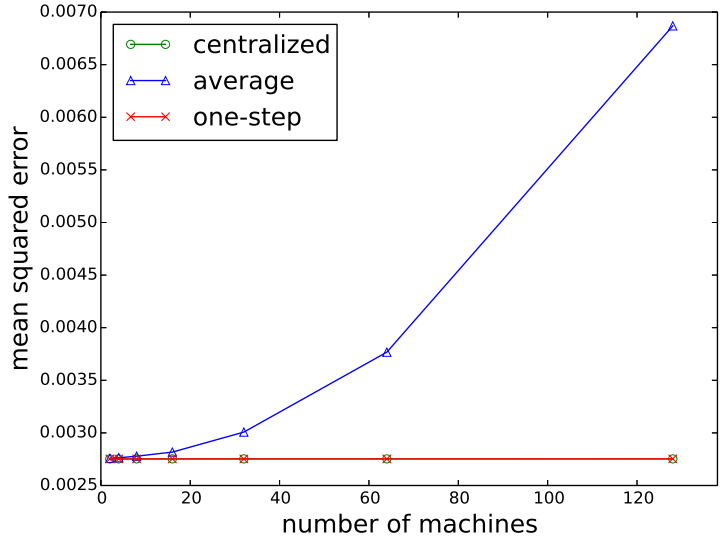}
		\caption{$d=20$}
	\end{subfigure}
	\begin{subfigure}[b]{0.6\textwidth}
		\centering
		\includegraphics[width=\textwidth]{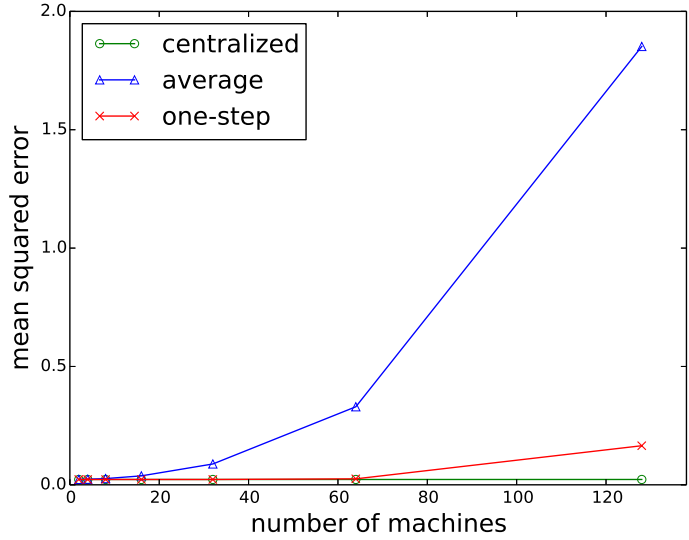}
		\caption{$d=100$}
	\end{subfigure}
	\caption{Logistic Regression: The mean squared error $\| \hat{\theta} - \theta_0 \|^2$ versus number of machines, with fifty simulations. The ``average" is $\theta^{(0)}$ and the ``one-step" is $\theta^{(1)}$. The ``centralized" denotes the oracle estimator with entire data.}
	\label{logistic}
\end{figure}
Detailed values of mean squared error are listed in Table \ref{logistic_table_low} and \ref{logistic_table_high}. From the tables, we can easily figure out that the standard deviation of the error of one-step estimator is significantly smaller than that of simple averaging, especially when the number of machines $k$ is large, which means one-step estimator is more stable.
\begin{table}[ht!]
\centering
%\hspace*{-2cm}
\begin{tabular}{| r | c | c | c | c | c | c | c |}
\hline
no. of machines & 2 & 4 & 8 & 16 & 32 & 64 & 128 \\
\hline
simple avg & 28.036 & 28.066 & 28.247 & 28.865 & 30.587 & 38.478 & 69.898 \\
($\times 10^{-4}$) & (7.982) & (7.989) & (8.145) & (8.443) & (9.812) & (14.247) & (27.655) \\
\hline
one-step & 28.038 & 28.038 & 28.038 & 28.038 & 28.038 & 28.035 & 28.039 \\
($\times 10^{-4}$) & (7.996) & (7.996) & (7.996) & (7.996) & (7.996) & (7.998) & (8.017) \\
\hline
centralized ($\times 10^{-4}$) & \multicolumn{7}{| c |}{28.038 (7.996)} \\
\hline
\end{tabular}
\caption{Logistic Regression ($d=20$): Detailed values of squared error $\| \hat{\theta} - \theta_0 \|^2$. In each cell, the first number is the mean of squared error in $K=50$ experiments and the number in the brackets is the standard deviation of the squared error.}
\label{logistic_table_low}
\end{table}

\begin{table}[ht!]
\centering
%\hspace*{-2cm}
\begin{tabular}{| r | c | c | c | c | c | c | c |}
\hline
no. of machines & 2 & 4 & 8 & 16 & 32 & 64 & 128 \\
\hline
simple avg & 23.066 & 23.818 & 26.907 & 38.484 & 87.896 & 322.274 & 1796.147 \\
($\times 10^{-3}$) & (4.299) & (4.789) & (6.461) & (10.692) & (22.782) & (67.489) & (324.274) \\
\hline
one-step & 22.787 & 22.784 & 22.772 & 22.725 & 22.612 & 24.589 & 151.440 \\
($\times 10^{-3}$) & (4.062) & (4.060) & (4.048) & (3.998) & (3.835) & (4.651) & (43.745) \\
\hline
centralized ($\times 10^{-3}$) & \multicolumn{7}{| c |}{22.787 (4.063)} \\
\hline
\end{tabular}
\caption{Logistic Regression ($d=100$): Detailed values of squared error $\| \hat{\theta} - \theta_0 \|^2$. In each cell, the first number is the mean of squared error in $K=50$ experiments and the number in the brackets is the standard deviation of squared error.}
\label{logistic_table_high}
\end{table}

\subsection{Beta Distribution}
In this example, we use data simulated from Beta distribution $\mbox{Beta}(\alpha,\beta)$, whose p.d.f. is as follows:
$$
f(x;\alpha,\beta) = \frac{\Gamma(\alpha+\beta)}{\Gamma(\alpha)\Gamma(\beta)}x^{\alpha-1}(1-x)^{\beta-1}.
$$

In each experiment, we generate the value of parameter as $\alpha \sim \mbox{Unif}(1,3)$ and $\beta \sim \mbox{Unif}(1,3)$, independently. Once $(\alpha,\beta)$ is determined, we can simulate samples from the above density. In order to examine the performance of two distributed methods when $k$ is extremely large, we choose to use a data set with relatively small size $N=2^{13}=8192$ and let number of machines vary in a larger range $k=2, 4, 8, \ldots, 256$. And the objective is to estimate parameter $(\alpha,\beta)$ from the observed data. The experiment is again repeated for $K=50$ times. The criterion function is $m(x;\theta) = \log f(x;\alpha,\beta)$, which implies that the centralized estimator is the MLE.

Figure \ref{betacomp} and Table \ref{beta_table} show that the one-step estimator has almost the same performance with centralized estimator in terms of MSE and standard deviation when the number of machines $k \le \sqrt{N}$ (i.e., when $k\le64$). However, the one-step estimator performs worse than centralized estimator when $k > \sqrt{N}$ (i.e., when $k=128$ or $256$), which confirms the necessity of condition $k=O(\sqrt{N})$ in Theorem \ref{one_step_M_ap}. In addition, we can easily find out that both simple averaging estimator and resampled averaging estimator are worse than the proposed one-step estimator regardless of the value of $k$.
\begin{figure}[h!]
	\centering
	\begin{subfigure}[b]{0.6\textwidth}
		\centering
		\includegraphics[width=\textwidth]{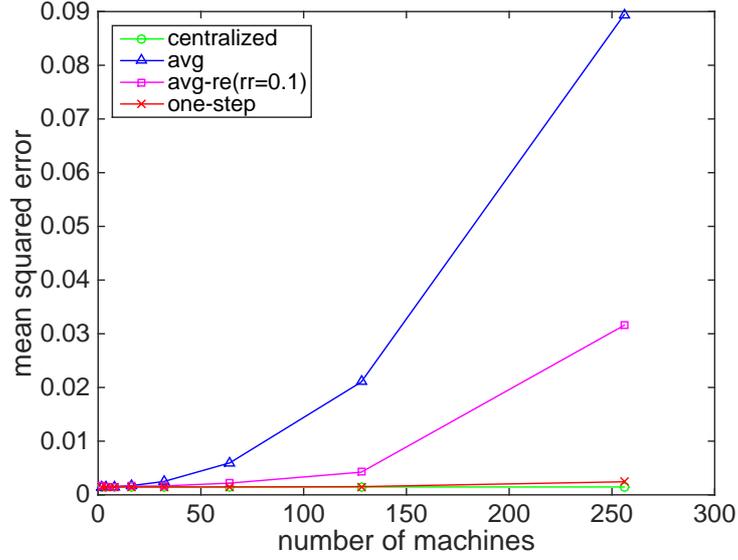}
		\caption{overview}
	\end{subfigure}
	\begin{subfigure}[b]{0.6\textwidth}
		\centering
		\includegraphics[width=\textwidth]{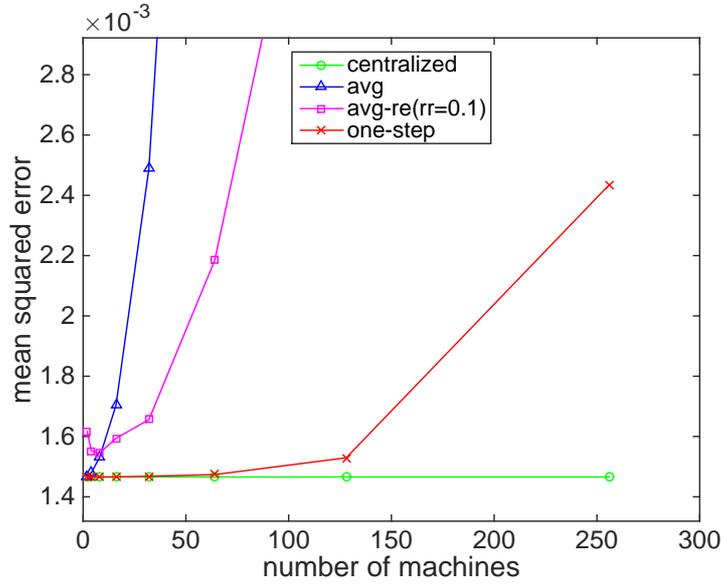}
		\caption{detailed view}
	\end{subfigure}
	\caption{Beta Distribution: The error $\| \theta - \theta_0 \|^2$ versus the number of machines, with fifty simulations, where $\theta_0$ is the true parameter. The ``avg" is $\theta^{(0)}$, the ``avg-re" is $\theta^{(0)}_{re}$ with resampling ratio $rr=10\%$ and the ``one-step" is $\theta^{(1)}$. The ``centralized" denotes maximum likelihood estimator with the entire data.}
	\label{betacomp}
\end{figure}

\begin{table}[ht!]
%\hspace*{-1cm}
\centering
\begin{tabular}{| r | c | c | c | c |}
\hline
\begin{tabular}{c} number of \\ machines \end{tabular} & \begin{tabular}{c} simple avg \\ ($\times 10^{-3}$) \end{tabular} & \begin{tabular}{c} resampled avg \\ ($\times 10^{-3}$) \end{tabular} & \begin{tabular}{c} one-step \\ ($\times 10^{-3}$) \end{tabular} & \begin{tabular}{c} centralized \\ ($\times 10^{-3}$) \end{tabular} \\
\hline
2 & 1.466 (1.936) & 1.616 (2.150) & 1.466 (1.943) & \multirow{9}{*}{\begin{tabular}{c} 1.466 \\ (1.943) \end{tabular}} \\ 
4 & 1.480 (1.907) & 1.552 (2.272) & 1.466 (1.943) & \\ 
8 & 1.530 (1.861) & 1.545 (2.177) & 1.466 (1.943) & \\ 
16 & 1.704 (1.876) & 1.594 (2.239) & 1.466 (1.946) & \\ 
32 & 2.488 (2.628) & 1.656 (2.411) & 1.468 (1.953) & \\ 
64 & 5.948 (5.019) & 2.184 (3.529) & 1.474 (1.994) & \\ 
128 & 21.002 (11.899) & 4.221 (7.198) & 1.529 (2.199) & \\ 
256 & 89.450 (35.928) & 31.574 (36.518) & 2.435 (3.384) & \\ 
\hline
\end{tabular}
\caption{Beta Distribution: Detailed values of squared error $\| \hat{\theta} - \theta_0 \|^2$. In each cell, the first number is the mean squared error with $K=50$ experiments and the number in the brackets is the standard deviation of the squared error.}
\label{beta_table}
\end{table}

\subsection{Beta Distribution with Possibility of Losing Information}
Now, we would like to compare the performance of simple averaging estimator and one-step estimator under a more practical scenario, in which each single local machine could fail to send its information to central machine. We assume those failures would occur independently with probability $r=0.05$. The simulation settings are similar to previous example in Section 4.2, however, we will generate $N=409600$ samples from Beta distribution $\mbox{Beta}(\alpha,\beta)$, where $\alpha$ and $\beta$ are chosen from $\mbox{Unif}(1,3)$, independently. And the goal of experiment is to estimate parameter $(\alpha,\beta)$. In each experiment, we let the number of machines vary $k=8,16, 32, 64, 128, 256, 512$. We also compare the performance of the centralized estimator with entire data and centralized estimator with $(1-r)\times100\% = 95\%$ of entire data. This experiment is repeated for $K=50$ times.

In Figure \ref{beta_crash} (a), we plot the MSE of each estimator against the number of machines. As expected, the MSE of simple averaging estimator grows significantly with the number of machines while the other three remains nearly the same. We can easily find out that performance of simple averaging estimator is far worse than others, especially when the number of machines is large (for instance, when $k = 256 \text{ or } 512$). If we take a closer look at the other three estimators from Fig. \ref{beta_crash} (b), we will find that the performance of one-step estimator is volatile but always remains in a reasonable range. And as expected, the error of one-step estimator converges to the error of oracle estimator with partial data when number of machines $k$ is large.
\begin{figure}[h!]
	\centering
	\begin{subfigure}[b]{0.6\textwidth}
		\centering
		\includegraphics[width=\textwidth]{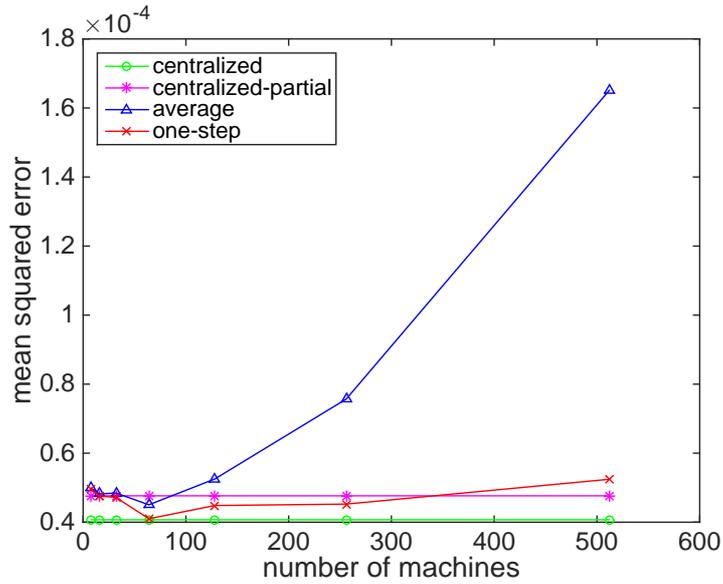}
		\caption{overview}
	\end{subfigure}
	\begin{subfigure}[b]{0.6\textwidth}
		\centering
		\includegraphics[width=\textwidth]{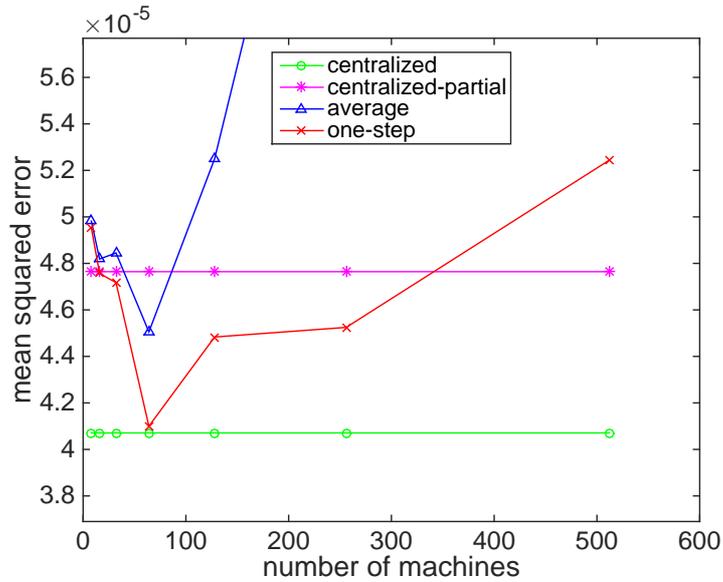}
		\caption{detailed view}
	\end{subfigure}
	\caption{Beta Distribution with Possibility of Losing Information: The error $\| \theta - \theta_0 \|^2$ versus the number of machines, with fifty simulations, where $\theta_0$ is the true parameter. The ``average" is $\theta^{(0)}$ and the ``one-step" is $\theta^{(1)}$. The ``centralized" denotes maximum likelihood estimator with the entire data. And the ``centralized-partial" denotes the maximum likelihood estimator with $(1-r)\times100\% = 95\%$ of data.}
	\label{beta_crash}
\end{figure}

\begin{table}[ht!]
\centering
\begin{tabular}{| r | c | c | c | c |}
\hline
\begin{tabular}{c} number of \\ machines \end{tabular} & \begin{tabular}{c} simple avg \\ ($\times 10^{-5}$) \end{tabular} & \begin{tabular}{c} one-step \\ ($\times 10^{-5}$) \end{tabular} & \begin{tabular}{c} centralized \\ ($\times 10^{-5}$) \end{tabular} & \begin{tabular}{c} centralized(95\%) \\ ($\times 10^{-5}$) \end{tabular} \\
\hline
8 & 4.98 (10.76) & 4.95 (10.62) & \multirow{7}{*}{\begin{tabular}{c} 4.07 \\ (6.91) \end{tabular}} & \multirow{7}{*}{\begin{tabular}{c} 4.76 \\ (9.81) \end{tabular}}\\ 
16 & 4.82 (7.61) & 4.75 (7.40) & &\\ 
32 & 4.85 (9.65) & 4.72 (9.31) & &\\ 
64 & 4.51 (7.89) & 4.10 (7.04) & &\\ 
128 & 5.25 (9.16) & 4.48 (7.77) & &\\ 
256 & 7.57 (12.26) & 4.52 (7.70) & &\\ 
512 & 16.51 (20.15) & 5.24 (8.02) & &\\ 
\hline
\end{tabular}
\caption{Beta Distribution with Possibility of Losing Information: Detailed values of squared error $\| \hat{\theta} - \theta_0 \|^2$. In each cell, the first number is the mean of squared error in $K=50$ experiments and the number in the brackets is the standard deviation of squared error.}
\label{beta_crash_table}
\end{table}

\subsection{Gaussian Distribution with Unknown Mean and Variance}
In this part, we will compare the performance of the simple averaging estimator, the resampled averaging estimator and the one-step estimator when fixing the number of machines $k = \sqrt{N}$ and letting the value of $N$ increase. We draw $N$ samples from $N(\mu,\sigma^2)$, where $\mu \sim \mbox{Unif}(-2,2)$ and $\sigma^2 \sim \mbox{Unif}(0.25,9)$, independently. We let $N$ vary in $\{ 4^3,\ldots,4^{9} \}$ and repeat the experiment for $K=50$ times for each $N$. We choose the criterion function to be the log-likelihood function
$$
m(x;\mu,\sigma^2) = - \frac{(x-\mu)^2}{2\sigma^2} - \frac{1}{2} \log (2\pi) - \frac{1}{2} \log \sigma^2.
$$

Figure \ref{gaussian} and Table \ref{gaussian_table} show that one-step estimator is asymptotically efficient while simple averaging estimator is absolutely not. It is worth noting that the resampled averaging estimator is not asymptotic efficient though it is better than simple averaging estimator. When the number of samples $N$ is relatively small, the one-step estimator is worse than centralized estimator. When the number of samples $N$ grows large, the differences between the one-step estimator and the centralized estimator become minimal in terms of both mean squared error and standard deviation. However, the error of the simple averaging estimator is significant larger than both the one-step estimator and the centralized estimator. When the sample size $N=4^{9}\approx 250,000$, the mean squared error of the simple averaging estimator is more than twice of that of the one-step and the centralized estimator.
\begin{figure}[h!]
	\centering
	\includegraphics[width=0.8\textwidth]{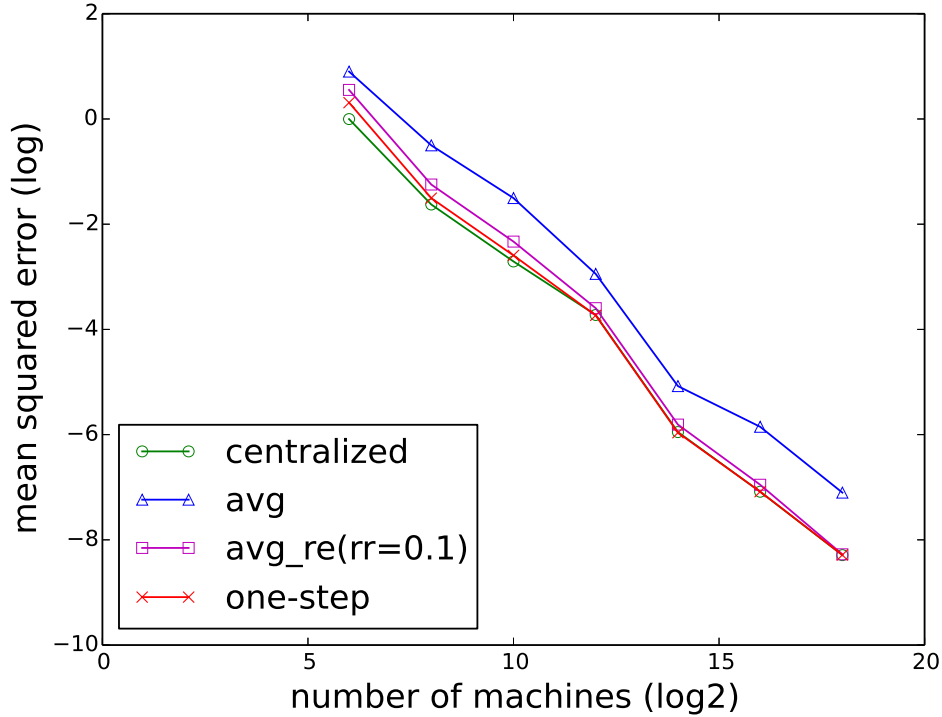}
	\caption{Gaussian Distribution with Unknown Mean and Variance: The log error $\log \| \theta - \theta_0 \|^2$ versus the log number of machines ($\log_2 k$), with fifty repeated experiments for each $N$, where $\theta_0$ is the true parameter. The ``avg", ``avg-re" and ``one-step" denote $\theta^{(0)}$, $\theta^{(0)}_{re}$ with resampling ratio $rr=10\%$ and $\theta^{(1)}$, respectively. The ``centralized" denotes the maximum likelihood estimator with the entire data. The sample size is fixed to be $N=k^2$.}
	\label{gaussian}
\end{figure}

\begin{table}[ht!]
%\hspace*{0cm}
\centering
\begin{tabular}{| r  r  c  c  c  c |}
\hline
\begin{tabular}{c} no. of \\ machines \end{tabular} & \begin{tabular}{c} no. of \\ samples \end{tabular} & simple avg &  resampled avg &  one-step & centralized \\
\hline
8 & 64 & \begin{tabular}{c} 3.022104 \\ (4.385627) \end{tabular} & \begin{tabular}{c} 2.153958 \\ (3.458645) \end{tabular} & \begin{tabular}{c} 1.694668 \\ (2.882794) \end{tabular} & \begin{tabular}{c} 1.388959 \\ (2.424813) \end{tabular} \\  \hline
16 & 256 & \begin{tabular}{c} 0.739784 \\ (1.209734) \end{tabular} & \begin{tabular}{c} 0.392389 \\ (0.739390) \end{tabular} & \begin{tabular}{c} 0.318765 \\ (0.621990) \end{tabular} & \begin{tabular}{c} 0.286175 \\ (0.566140) \end{tabular} \\ \hline
32 & 1024 & \begin{tabular}{c} 0.118766 \\ (0.151695) \end{tabular} & \begin{tabular}{c} 0.041050 \\ (0.053808) \end{tabular} & \begin{tabular}{c} 0.034494 \\ (0.046586) \end{tabular} & \begin{tabular}{c} 0.032563 \\ (0.045779) \end{tabular} \\ \hline
64 & 4096 & \begin{tabular}{c} 0.026839 \\ (0.046612) \end{tabular} & \begin{tabular}{c} 0.016519 \\ (0.030837) \end{tabular} & \begin{tabular}{c} 0.014255 \\ (0.029258) \end{tabular} & \begin{tabular}{c} 0.014414 \\ (0.030533) \end{tabular} \\ \hline
128 & 16384 & \begin{tabular}{c} 0.010996 \\ (0.019823) \end{tabular} & \begin{tabular}{c} 0.004542 \\ (0.009089) \end{tabular} & \begin{tabular}{c} 0.004329 \\ (0.009453) \end{tabular} & \begin{tabular}{c} 0.004357 \\ (0.009315) \end{tabular} \\ \hline
256 & 65536 & \begin{tabular}{c} 0.002909 \\ (0.005785) \end{tabular} & \begin{tabular}{c} 0.001158 \\ (0.002733) \end{tabular} & \begin{tabular}{c} 0.001105 \\ (0.002779) \end{tabular} & \begin{tabular}{c} 0.001099 \\ (0.002754) \end{tabular} \\ \hline
512 & 262144 & \begin{tabular}{c} 0.000843 \\ (0.001426) \end{tabular} & \begin{tabular}{c} 0.000461 \\ (0.000744) \end{tabular} & \begin{tabular}{c} 0.000376 \\ (0.000596) \end{tabular} & \begin{tabular}{c} 0.000376 \\ (0.000595) \end{tabular} \\
\hline
\end{tabular}
\caption{Gaussian Distribution with Unknown Mean and Variance: Detailed values of squared error $\| \hat{\theta} - \theta_0 \|^2$. In each cell, the first number is the mean of squared error in $K=50$ experiments and the number in the brackets is the standard deviation of squared error.}
\label{gaussian_table}
\end{table}

\section{Conclusion}
\label{conclusion}

The M-estimator is a fundamental and high-impact methodology in statistics.
The classic M-estimator theory is based on the assumption that the entire data are available at a central location, and can be processed/computed without considering communication issues.
In many modern estimation problems arising in contemporary sciences and engineering, the classical notion of asymptotic optimality suffers from a significant deficiency: it requires access to all data.
The asymptotic property when the data has to be dealt with distributively is under-developed.
In this paper, we close this gap by considering a distributed one-step estimator.

Our one-step estimator builds on the existing {\it averaging} estimator.
In a nutshell, after obtaining an averaging estimator, this initial estimate is broadcasted to local machines, to facilitate their computation of gradients and hessians of their objective functions.
By doing so, the data do {\it not} need to be transmitted to the central machine.
The central machine than collects the locally estimated gradients and hessians, to produce a global estimate of the overall gradient and overall hessian.
Consequently, a one-step update of the initial estimator is implemented.
Just like the one-step approach has improved the estimator in the classical (non-distributed) setting, we found that the one-step approach can improve the performance of an estimator under the distributed setting, both theoretically and numerically.

Besides the works that have been cited earlier, there are many other results that are in the relevant literature, however they may not be directly technically linked to what's been done here.
We discuss their influence and insights in the next few paragraphs.

An interesting split-and-merge Bayesian approach for variable selection under linear models is proposed in \cite{LiangJRSSB2015}.
The method firstly split the ultrahigh dimensional data set into a number of lower dimensional subsets
and select relevant variables from each of the subsets, and then aggregate the variables selected
from each subset and then select relevant variables from the aggregated data set.
Under mild conditions, the authors show that the proposed approach is consistent, i.e., the underlying true model will be selected in probability $1$ as the sample size becomes large.
This work differs from all the other approaches that we discussed in this paper: it splits the variables, while all other approaches that we referenced (including ours) split the data according to observations.
This paper certainly is in line with our research, however takes a very distinct angle.

An interesting piece of work that combines distributed statistical inference and information theory in communication is presented in \cite{Zhang2013NIPS}.
Their current results need to rely on special model settings: uniform location family $\mathcal{U} = \{P_{\theta}, \theta \in [-1, 1]\}$, where $P_\theta$ denotes the uniform distribution on the interval $[\theta - 1, \theta + 1]$, or Gaussian location families $N_d([-1, 1]^d) = \{N(\theta, \sigma^2 I_{d\times d}) \mid \theta \in \Theta = [-1, 1]^d \}$.
It will be interesting to see whether or not more general results are feasible.

\cite{nowak2003distributed} proposed a distributed expectation-maximization (EM) algorithm for density estimation and clustering in sensor networks.
Though the studied problem is technically different from ours, it provides an inspiring historic perspective: distributed inference has been studied more than ten years ago.

\cite{neiswanger2013asymptotically} propose an asymptotically exact, embarrassingly parallel MCMC method by approximating each sub-posterior with Gaussian density, Gaussian kernel or weighted Gaussian kernel.
They prove the asymptotic correctness of their estimators and bound rate of convergence.
Our paper does not consider the MCMC framework.
The analytical tools that they used in proving their theorems are of interests.

\cite{wang2014median} propose a distributed variable selection algorithm, which accepts a variable if more than half of machines select that variable.
They give upper bounds for the success probability and Mean Squared Error (MSE) of estimator.
This work bears similarity with \cite{LiangJRSSB2015} and \cite{chen2014split}, however with somewhat different emphases.

\cite{kleiner2014scalable} propose a scalable bootstrap (named `bag of little bootstraps' (BLB)) for massive data to assess the quality of estimators.
They also demonstrate its favorable statistical performance through both theoretical analysis and simulation studies.
A comparison with this work will be interesting, however not included here.

\cite{zhao2014partially} consider a partially linear framework for massive heterogeneous data and propose an aggregation type estimator for the commonality parameter that possesses the minimax optimal bound and asymptotic distribution when number of sub-populations does not grow too fast.

A recent work \citep{Arjevani2015} shed interesting new light into the distributed inference problem.
The authors studied the fundamental limits to communication-efficient distributed methods for convex learning and optimization, under different assumptions on the information available to individual machines, and the types of functions considered.
The current problem formulation is more numerical than statistical properties.
Their idea may lead to interesting counterparts in statistical inference.

Besides estimation, other distributed statistical technique may be of interests, such as the distributed principal component analysis \cite{distributePCA2014}.
This paper does not touch this line of research.

Various researchers have studied communication-efficient algorithms for statistical estimation (e.g., see the papers \cite{Ofer2012,Balcan2012,Wainwright2014ICM,McDonald2010} and references therein).
They were not discussed in details here, because they are pretty much discussed/compared in other references of this paper.

There is now a rich and well-developed body of theory for bounding and/or computing the minimax risk for various statistical estimation problems, e.g., see \cite{Yang1999Aos} and references therein.
In several cited references, researchers have started to derive the optimal minimax rate for estimators under the distributed inference setting.
This will be an exciting future research direction.

%%%%%%%%%%%%%%%%%%%%%%%%%%%%%%%%%%%%%%%%%%
\bibliographystyle{plainnat}
\bibliography{dMLEbib}

\appendix
\renewcommand{\theequation}{\Alph{section}.\arabic{equation}}

\section*{Appendix Overview}
The appendix is organized as follows. In Section \ref{bounds}, we analyze the upper bounds of sum of i.i.d. random vectors and random matrices, which will be useful in later proofs. In Section \ref{error_bound}, we derive the upper bounds of the local M-estimators and the simple averaging estimator. We present the proofs to Theorem \ref{one_step_M_ap} and Theorem \ref{one_step_M} in Section \ref{proof_ap} and Section \ref{proof}, respectively. A proof of Corollary \ref{estimation_with_crash} will be in Section \ref{proof_crash}.

\section{Bounds on Gradient and Hessian}
\label{bounds}
In order to establish the convergence of gradients and Hessians of the empirical criterion function to those of population criterion function, which is essential for the later proofs, we will present some results on the upper bound of sums of i.i.d. random vectors and random matrices. We start with stating a useful inequality on the sum of independent random variables from \cite{rosenthal1970subspaces}.
\begin{lemma}[Rosenthal's Inequality, \cite{rosenthal1970subspaces}, Theorem 3]
\label{Rosenthal}
For $q > 2$, there exists constant $C(q)$ depending only on $q$ such that if $X_1,\ldots,X_n$ are independent random variables with $\mathbb{E}[X_j] = 0$ and $\mathbb{E}[|X_j|^q] < \infty$ for all $j$, then
$$
(\mathbb{E} [| \sum_{j=1}^n X_j |^q])^{1/q} \le C(q) \max \left\{ (\sum_{j=1}^n \mathbb{E}[|X_j|^q])^{1/q} , (\sum_{j=1}^n \mathbb{E}[ |X_j|^2 ])^{1/2} \right\}
$$
\end{lemma}

Equipped with the above lemma, we can bound the moments of mean of random vectors.
\begin{lemma}
\label{vector_bound}
Let $X_1, \ldots, X_n \in \mathbb{R}^d$ be i.i.d. random vectors with $\mathbb{E}[X_i] = \mathbf{0}$. And there exists some constants $G>0$ and $q_0 \ge 2$ such that $\mathbb{E} [\| X_i \|^{q_0}] < G^{q_0}$. Let $\mean{X} = \frac{1}{n}\sum_{i=1}^n X_i$, then for $1 \le q \le q_0$, we have
$$
\mathbb{E}[\| \mean{X} \|^q] \le \frac{C_v(q,d)}{n^{q/2}}G^q,
$$
where $C(q,d)$ is a constant depending solely on $q$ and $d$.
\end{lemma}
\begin{proof}
The main idea of this proof is to transform the sum of random vectors into the sum of random variables and then apply Lemma \ref{Rosenthal}. Let $X_{i,j}$ denote the $j$-th component of $X_i$ and $\mean{X}_j$ denote the $j$-th component of $\mean{X}$. 
\begin{enumerate}
\item Let us start with a simpler case in which $q=2$. 
\begin{align*}
\mathbb{E} [\| \mean{X} \|^2] &= \sum_{j=1}^d \mathbb{E} [ |\mean{X}_j|^2] = \sum_{j=1}^d \sum_{i=1}^n \mathbb{E}[| X_{i,j}/n |^2] \\
&= \sum_{j=1}^d \mathbb{E}[|X_{1,j}|^2]/n = n^{-1} \mathbb{E}[\|X_1\|^2] \le n^{-1}G^2,
\end{align*}
The last inequality holds because $\mathbb{E}[\|X_{1}\|^q] \le (\mathbb{E}[\|X_{1}\|^{q_0}])^{q/q_0} \le G^q$ for $1 \le q \le q_0$ by H\"{o}lder's inequality.
\item When $1 \le q < 2$, we have
$$
\mathbb{E} [\| \mean{X} \|^q] \le (\mathbb{E} [\| \mean{X} \|^2])^{q/2} \le n^{-q/2}G^q.
$$
\item For $2 < q \le q_0$, with some simple algebra, we have
\begin{align*}
\mathbb{E} [\| \mean{X} \|^q] &= \mathbb{E} \left[ (\sum_{j=1}^d |\mean{X}_j|^2)^{q/2} \right] \\
&\le \mathbb{E} \left[ (d \max_{1 \le j \le d} |\mean{X}_j|^2)^{q/2} \right] = d^{q/2} \mathbb{E} \left[ \max_{1 \le j \le d} |\mean{X}_j|^q \right] \\
&\le d^{q/2} \mathbb{E} \left[ \sum_{j=1}^d |\mean{X}_j|^q \right] = d^{q/2} \sum_{j=1}^d \mathbb{E} \left[ |\mean{X}_j|^q \right].
\end{align*}
As a continuation, we have
\begin{align*}
& \quad \mathbb{E} [\| \mean{X} \|^q] \le d^{q/2} \sum_{j=1}^d \mathbb{E}[|\mean{X}_j|^q] \\
&\le d^{q/2} \sum_{j=1}^d [C(q)]^q \max \left\{ \sum_{i=1}^n \mathbb{E}[|X_{i,j}/n|^q] , (\sum_{i=1}^n \mathbb{E}[ |X_{i,j}/n|^2 ])^{q/2} \right\} \qquad \text{(Lemma \ref{Rosenthal})}\\
&= d^{q/2} [C(q)]^q \sum_{j=1}^d \max \left\{ \frac{\mathbb{E}[|X_{1,j}|^q]}{n^{q-1}} , \frac{(\mathbb{E}[|X_{1,j}|^2])^{q/2}}{n^{q/2}} \right\} \\
&\le d^{q/2+1} [C(q)]^q \max \left\{ \frac{\mathbb{E}[\|X_{1}\|^q]}{n^{q-1}} , \frac{(\mathbb{E}[\|X_{1}\|^2])^{q/2}}{n^{q/2}} \right\} \qquad \text{(since $\mathbb{E}[|X_{1,j}|^q] \le \mathbb{E}[\|X_{1}\|^q]$)}\\
&\le d^{q/2+1} [C(q)]^q \max \left\{ \frac{G^q}{n^{q-1}} , \frac{G^q}{n^{q/2}} \right\} \qquad \text{(H\"{o}lder's inequality)}\\
&= \frac{d^{q/2+1} [C(q)]^q}{n^{q/2}}G^q. \qquad \text{($q-1>q/2$ when $q>2$)}
\end{align*}
To complete this proof, we just need to set $C_v(q,d)=d^{q/2+1} [C(q)]^q$.
\end{enumerate}
\end{proof}

To bound the moment of the mean of i.i.d. random matrices, let us consider another matrix norm -- Frobenius norm $\vertiii{\cdot}_F$, i.e., 
$$\vertiii{A}_F = \sqrt{\sum_{i,j} | a_{ij} |^2}, \forall A \in \mathbb{R}^{d \times d}.$$
Note that
$$
\vertiii{A}_F = \sqrt{\sum_{i,j} | a_{ij} |^2} = \sqrt{\mbox{trace} (A^tA)} \ge \sqrt{ \sup_{u \in \mathbb{R}^d: \|u\| \le 1} \| A^tAu \| } = \vertiii{A},
$$
and
$$
\vertiii{A}_F \le \sqrt{ d \sup_{u \in \mathbb{R}^d: \|u\| \le 1} \| A^tAu \| } = \sqrt{d}\, \vertiii{A}.
$$
With Frobenius norm, we can regard a random matrix $X \in \mathbb{R}^{d \times d}$ as a random vector in $\mathbb{R}^{d^2}$ and apply Lemma \ref{vector_bound} to obtain the following lemma.
\begin{lemma}
\label{matrix_bound}
Let $X_1,\ldots,X_n \in \mathbb{R}^{d \times d}$ be i.i.d. random matrices with $\mathbb{E}[X_i] = \mathbf{0}_{d \times d}$. Let $\vertiii{X_i}$ denote the norm of $X_i$, which is defined as its maximal singular value. Suppose $\mathbb{E} [\vertiii{X_i}^{q_0}] \le H^{q_0}$, where $q_0 \ge 2$ and $H>0$. Then for $\mean{X} = \frac{1}{n}\sum_{i=1}^n X_i$ and $1 \le q \le q_0$, we have
$$
\mathbb{E} \left[ \vertiii{\mean{X}}^q \right] \le \frac{C_m(q,d)}{n^{q/2}}H^q,
$$
where $C_m(q,d)$ is a constant depending on $q$ and $d$ only.
\end{lemma}
\begin{proof}
By the fact $\vertiii{A}_F \le \sqrt{d} \vertiii{A}$, we have
$$
\mathbb{E} \left[ \vertiii{X_i}_F^{q_0} \right] \le \mathbb{E} \left[ \vertiii{\sqrt{d} X_i}^{q_0} \right] \le (\sqrt{d}H)^{q_0}.
$$
Then by the fact $ \vertiii{A} \le \vertiii{A}_F$ and Lemma \ref{vector_bound}, we have
\begin{align*}
\mathbb{E} \left[ \vertiii{\mean{X}}^q \right] \le \mathbb{E} \left[ \vertiii{\mean{X}}_F^q \right] \le \frac{C_v(q,d^2)}{n^{q/2}}(\sqrt{d}H)^q = \frac{C_v(q,d^2)d^\frac{q}{2}}{n^{q/2}}H^q.
\end{align*}
In the second inequality, we treat $\mean{X}$ as a $d^2$-dimensional random vector and then apply Lemma \ref{vector_bound}. Then the proof can be completed by setting $C_m(q,d) = C_v(q,d^2)d^\frac{q}{2}$.
\end{proof}

\section{Error Bound of Local M-estimator and Simple Averaging Estimator}
\label{error_bound}
Since the simple averaging estimator is the average of all local estimators and the one-step estimator is just a single Newton-Raphson update from the simple averaging estimator. Thus, it is natural to study the upper bound of the mean squared error (MSE) of a local M-estimator and the upper bound of the MSE of the simple averaging estimator. The main idea in the following proof is similar to the thread in the proof of \textbf{Theorem 1} in \cite{zhang2013communication}, but the conclusions are different. Besides, in the following proof, we use a correct analogy of mean value theorem for vector-valued functions.
\subsection{Bound the Error of Local M-estimators $\theta_i, i=1,\ldots,k$}
In this subsection, we would like to analyze the mean squared error of a local estimator $\theta_i = \argmax_{\theta \in \Theta} M_i(\theta)$, $i = 1, \ldots, k$ and prove the following lemma in the rest of this subsection.
\begin{lemma}
\label{asym_norm2}
Let $\Sigma = \ddot{M}_0(\theta_0)^{-1} \mathbb{E}[\dot{m}(X;\theta_0)\dot{m}(X;\theta_0)^t] \ddot{M}_0(\theta_0)^{-1}$, where the expecation is taken with respect to $X$. Under Assumption \ref{parameter_space}, \ref{invertibility} and \ref{smoothness}, for each $i=1,\ldots,k$, we have
$$
\mathbb{E} [\| \theta_i - \theta_0 \|^2] \le \frac{2}{n} Tr(\Sigma) + O(n^{-2}).
$$
\end{lemma}
\noindent
Since $\dot{M}_i(\theta_i) = 0$, by Theorem 4.2 in Chapter XIII of \cite{lang1993real}, we have
\begin{align*}
0 &= \dot{M}_i(\theta_i) = \dot{M}_i(\theta_0) + \int_{0}^1 \ddot{M}_i((1 - \rho) \theta_0 +\rho \theta_i) d\rho \, [\theta_i - \theta_0] \\
&= \dot{M}_i(\theta_0) + \ddot{M}_0(\theta_0) [\theta_i - \theta_0] + \left[ \int_{0}^1 \ddot{M}_i((1 - \rho) \theta_0 +\rho \theta_i) d\rho - \ddot{M}_0(\theta_0) \right] [\theta_i - \theta_0] \\
%& \qquad \text{(add and subtract $\ddot{M}_0(\theta_0)$)}\\
&= \dot{M}_i(\theta_0) + \ddot{M}_0(\theta_0) [\theta_i - \theta_0] + \left[ \int_{0}^1 \ddot{M}_i((1 - \rho) \theta_0 +\rho \theta_i) d\rho \, - \ddot{M}_i(\theta_0) \right] [\theta_i - \theta_0] \\
& \qquad + [\ddot{M}_i(\theta_0) - \ddot{M}_0(\theta_0)] [\theta_i - \theta_0] \qquad \text{(subtract and add $\ddot{M}_i(\theta_0)$)},
\end{align*}
\begin{remark*}
Here, it is worth noting that there is no analogy of mean value theorem for vector-valued functions, which implies that there does not necessarily exist $\theta'$ lying on the line between $\theta_i$ and $\theta_0$ satisfying $\dot{M}_i(\theta_i) - \dot{M}_i(\theta_0) = \ddot{M}_i(\theta') (\theta_i - \theta_0)$. Numerous papers make errors by claiming such $\theta'$ lies between $\theta_i$ and $\theta_0$.
\end{remark*}
If last two terms in above equation are reasonably small, this lemma follows immediately. So, our strategy is as follows. First, we show that the mean squared error of both $[\int_{0}^1 \ddot{M}_i((1 - \rho) \theta_0 +\rho \theta_i) d\rho \, - \ddot{M}_i(\theta_0)] [\theta_i - \theta_0]$ and $[\ddot{M}_i(\theta_0) - \ddot{M}_0(\theta_0)] [\theta_i - \theta_0]$ is small under some ``good" events. Then we will show the probability of ``bad" events is small enough. And Lemma \ref{asym_norm2} will follow by the fact that $\Theta$ is compact.

Suppose $S_i = \{x_1,\ldots,x_n\}$ is the data set on local machine $i$. Let us define some good events:
\begin{eqnarray*}
E_1 &=& \left\{ \frac{1}{n}\sum_{j=1}^n L(x_j) \le 2L \right\}, \\
E_2 &=& \left\{ \vertiii{ \ddot{M}_i(\theta_0) - \ddot{M}_0(\theta_0)} \le \lambda/4 \right\}, \\
E_3 &=& \left\{ \| \dot{M}_i(\theta_0) \| \le \frac{\lambda}{2} \delta' \right\},
\end{eqnarray*}
where $\delta' = \min(\delta,\frac{\lambda}{8L})$, $\lambda$ is the constant in Assumption \ref{invertibility} and $L$ and $\delta$ are the constants in Assumption \ref{smoothness}. We will show that event $E_1$ and $E_2$ ensure that $M_i(\theta)$ is strictly concave at a neighborhood of $\theta_0$. And we will also show that in event $E_3$, $\theta_i$ is fairly close to $\theta_0$. Let $E = E_1 \cap E_2 \cap E_3$, then we have the following lemma:
\begin{lemma}
Under event $E$, we have
$$
\| \theta_i - \theta_0 \| \le \frac{4}{\lambda}  \| \dot{M}_i(\theta_0) \|
$$
\end{lemma}

\begin{proof}
First, we will show $\ddot{M}_i(\theta)$ is a negative definite matrix over a ball centered at $\theta_0$: $B_{\delta'} = \{ \theta \in \Theta: \| \theta - \theta_0 \| \le \delta' \} \subset B_{\delta}$. For any fixed $\theta \in B_{\delta'}$, we have
\begin{eqnarray*}
\vertiii{ \ddot{M}_i(\theta) - \ddot{M}_0(\theta_0) } &\le& \vertiii{\ddot{M}_i(\theta) - \ddot{M}_i(\theta_0)} + \vertiii{ \ddot{M}_i(\theta_0) - \ddot{M}_0(\theta_0) } \\
&\le& 2L\| \theta - \theta_0 \| + \frac{\lambda}{4} \le \lambda/4 + \lambda/4 = \lambda/2, 
\end{eqnarray*}
where we apply event $E_1$, Assumption \ref{smoothness} and the fact that $\delta'=\min(\delta,\frac{\lambda}{8L})$ on the first term and event $E_2$ on the second term. Since $\ddot{M}_0(\theta_0)$ is negative definite by Assumption \ref{invertibility}, above inequality implies that $\ddot{M}_i(\theta)$ is negative definite for all $\theta \in B_{\delta'}$ and
\begin{equation}
\label{concavity}
\sup_{u \in \mathbb{R}^d : \| u \| \le 1} u^t \ddot{M}_i(\theta) u \le -\lambda/2. 
\end{equation}
With negative definiteness of $\ddot{M}_i(\theta), \theta \in B_{\delta'}$, event $E_3$ and concavity of $M_i(\theta), \theta \in \Theta$, we have
$$
\frac{\lambda}{2} \delta' \stackrel{E_3}{\ge} \| \dot{M}_i(\theta_0) \| = \| \dot{M}_i(\theta_0) - \dot{M}_i(\theta_i) \|  \stackrel{(\ref{concavity})}{\ge} \frac{\lambda}{2} \| \theta_i - \theta_0 \|.
$$
Thus, we know $\| \theta_i - \theta_0 \| \le \delta'$, or equivalently, $\theta_i \in B_{\delta'}$. Then by applying Taylor's Theorem on $M_i(\theta)$ at $\theta_0$, we have
\begin{eqnarray*}
M_i(\theta_i) \stackrel{ (\ref{concavity}) }{\le} M_i(\theta_0) + \dot{M}_i(\theta_0)^t (\theta_i - \theta_0) - \frac{\lambda}{4} \| \theta_i - \theta_0 \|^2.
\end{eqnarray*}
Thus, as $M_i(\theta_0) \le M_i(\theta_i)$ by definiton,
\begin{align*}
\frac{\lambda}{4} \| \theta_i - \theta_0 \|^2 &\le M_i(\theta_0) - M_i(\theta_i) + \dot{M}_i(\theta_0)^t (\theta_i - \theta_0) \\
&\le  \| \dot{M}_i(\theta_0) \| \| \theta_i - \theta_0 \|,
\end{align*}
which implies
$$
\| \theta_i - \theta_0 \| \le \frac{4}{\lambda}  \| \dot{M}_i(\theta_0) \|.
$$
\end{proof}

For $1 \le q \le 8$, we can bound $\mathbb{E}[ \| \dot{M}_i(\theta_0) \|^q ]$ by Lemma \ref{vector_bound} and Assumption \ref{smoothness},
$$
\mathbb{E}[ \| \dot{M}_i(\theta_0) \|^q ] \le \frac{C_v(q,d)}{n^{q/2}}G^q,
$$
where $C_v(q,d)$ is a constant depending on $q$ and $d$ only. Then by conditioning on event $E$, we have
\begin{eqnarray*}
\mathbb{E}[ \| \theta_i - \theta_0 \|^q] &=& \mathbb{E}[ \| \theta_i - \theta_0 \|^q 1_{(E)}] + \mathbb{E}[ \| \theta_i - \theta_0 \|^q 1_{(E^c)}] \\
&\le& \frac{4^q}{\lambda^q}\mathbb{E}[ \| \dot{M}_i(\theta_0) \|^q ] + D^q \mbox{Pr}(E^c) \\
&\le& \frac{4^q}{\lambda^q}\frac{C_v(q,d)}{n^{q/2}}G^q + D^q \mbox{Pr}(E^c).
\end{eqnarray*}
If we can show $\mbox{Pr}(E^c) = O(n^{-\frac{q}{2}})$, then $\mathbb{E}[ \| \theta_i - \theta_0 \|^q] = O(n^{-\frac{q}{2}})$ follows immediately.

\begin{lemma}
\label{prob_good_event}
Under Assumption \ref{smoothness}, we have
$$
\mbox{Pr}(E^c) = O(n^{-4}).
$$
\end{lemma}
\begin{proof}
Under Assumption \ref{smoothness}, by applying Lemma \ref{vector_bound} and \ref{matrix_bound}, we can bound the moments of $\dot{M}_i(\theta_0)$ and $\ddot{M}_i(\theta_0) - \ddot{M}_0(\theta_0)$. Rigorously, for $1 \le q \le 8$, we have
\begin{eqnarray*}
\mathbb{E} \left[ \|\dot{M}_i(\theta_0)\|^q \right] &\le& \frac{C_v(q,d)}{n^{q/2}}G^q, \\
\mathbb{E} \left[ \vertiii{ \ddot{M}_i(\theta_0) - \ddot{M}_0(\theta_0)}^q \right] &\le& \frac{C_m(q,d)}{n^{q/2}}H^q.
\end{eqnarray*}
Therefore, by Markov's inequality, we have
\begin{eqnarray*}
& & \mbox{Pr}(E^c) = \mbox{Pr}(E_1^c \cup E_2^c  \cup E_3^c) \le \mbox{Pr}(E_1^c) + \mbox{Pr}(E_2^c) + \mbox{Pr}(E_3^c) \\
&\le& \frac{\mathbb{E} \left[ |\frac{1}{n}\sum_{j=1}^n L(x_j) - \mathbb{E}[L(x)]|^8 \right]}{L^8} + \frac{\mathbb{E} \left[ \vertiii{ \ddot{M}_i(\theta_0) - \ddot{M}_0(\theta_0)}^8 \right]}{(\lambda/4)^8} + \frac{\mathbb{E} \left[ \|\dot{M}_i(\theta_0)\|^8 \right]}{(\lambda \delta'/2)^8} \\
&\le& O(\frac{1}{n^4}) + O(\frac{1}{n^4}) + O(\frac{1}{n^4}) = O(n^{-4}).
\end{eqnarray*}
\end{proof}

Now, we have showed that for $1 \le q \le 8$,
\begin{equation}
\label{q_bound_of_theta_i}
\mathbb{E}[ \| \theta_i - \theta_0 \|^q] \le \frac{4^q}{\lambda^q}\frac{C_v(q,d)}{n^{q/2}}G^q + O(n^{-4}) = O(n^{-\frac{q}{2}}).
\end{equation}
Until now, $[\ddot{M}_i(\theta_0) - \ddot{M}_0(\theta_0)] [\theta_i - \theta_0]$ has been well bounded. Next, we will consider the moment bound of $\int_{0}^1 \ddot{M}_i((1 - \rho) \theta_0 +\rho \theta_i) d\rho - \ddot{M}_i(\theta_0)$.
\begin{lemma}
\label{q_bound_of_matrix}
Under assumption \ref{smoothness}, for $1 \le q \le 4$,
$$
\mathbb{E}\left[ \vertiii{ \int_{0}^1 \ddot{M}_i((1 - \rho) \theta_0 +\rho \theta_i) d\rho - \ddot{M}_i(\theta_0) }^q \right] \le L^q \frac{4^q}{\lambda^q}\frac{\sqrt{C_v(2q,d)}}{n^{q/2}}G^q + O(n^{-2}) = O(n^{-q/2}).
$$
\end{lemma}
\begin{proof}
By Minkowski's integral inequality, we have
\begin{align*}
\mathbb{E}\left[ \vertiii{ \int_{0}^1 \ddot{M}_i((1 - \rho) \theta_0 +\rho \theta_i) d\rho - \ddot{M}_i(\theta_0) }^q \right] &\le \mathbb{E}\left[ \int_{0}^1 \vertiii{ \ddot{M}_i((1 - \rho) \theta_0 +\rho \theta_i) - \ddot{M}_i(\theta_0) }^q d\rho \right] \\
&= \int_{0}^1 \mathbb{E}\left[ \vertiii{ \ddot{M}_i((1 - \rho) \theta_0 +\rho \theta_i) - \ddot{M}_i(\theta_0) }^q \right] d\rho.
\end{align*}
For simplicity of notation, we use $\theta' = (1 - \rho) \theta_0 +\rho \theta_i$ in this proof.
When event $E$ holds, we have
$$
\| \theta' - \theta_0 \| = \| \rho(\theta_i - \theta_0) \| \le \rho \delta' \le \delta,
$$
which means that $\theta' \in B_{\delta}, \forall \rho \in [0,1]$. Thus, because of the convexity of the matrix norm $\vertiii{\cdot}$, we can apply Jensen's inequality and Assumption \ref{smoothness} and get
\begin{eqnarray*}
\vertiii{ \ddot{M}_i(\theta') - \ddot{M}_i(\theta_0) }^q \le \frac{1}{n} \sum_{j=1}^n \vertiii{ \ddot{m}(x_j;\theta') - \ddot{m}(x_j;\theta_0)}^q \le \frac{1}{n} \sum_{j=1}^n L(x_i)^q \| \theta' - \theta_0 \|^q.
\end{eqnarray*}
Then apply H\"{o}lder's inequality,
\begin{eqnarray*}
\mathbb{E} \left[ \vertiii{ \ddot{M}_i(\theta') - \ddot{M}_i(\theta_0) }^q 1_{(E)} \right] &\le& \left\{ \mathbb{E}[ (\frac{1}{n} \sum_{j=1}^n L(x_i)^q)^2 ] \right\} ^{1/2} \left\{ \mathbb{E}[ \| \theta' - \theta_0 \|^{2q} ] \right\} ^{1/2} \\
&\stackrel{\mbox{Jensen's}}{\le}& C(q) L^q \rho^q \left\{ \mathbb{E}[ \| \theta_i - \theta_0 \|^{2q} ] \right\} ^{1/2} \\
&\stackrel{\mbox{(\ref{q_bound_of_theta_i})}}\le& C(q) L^q \frac{4^q \sqrt{C(2q,d)} G^q}{\lambda^q n^{q/2}} + O(n^{-2}).
\end{eqnarray*}
When event $E$ does not hold, we know that $\vertiii{ \ddot{M}_i(\theta') - \ddot{M}_i(\theta_0) }^q$ must be finite by the assumption that $\Theta$ is compact and $\ddot{M}_i(\theta)$ is continuous. By Lemma \ref{prob_good_event}, the probability that event $E$ does not hold is bounded by $O(n^{-4})$, which implies,
$$
\mathbb{E} \left[ \vertiii{ \ddot{M}_i(\theta') - \ddot{M}_i(\theta_0) }^q \right] \le C(q) L^q \frac{4^q \sqrt{C(2q,d)} G^q}{\lambda^q n^{q/2}} + O(n^{-2}) + O(n^{-4}) = O(n^{-q/2}).
$$
Therefore, we have
\begin{align*}
\mathbb{E}\left[ \vertiii{ \int_{0}^1 \ddot{M}_i((1 - \rho) \theta_0 +\rho \theta_i) d\rho - \ddot{M}_i(\theta_0) }^q \right] \le \int_{0}^1 \mathbb{E}\left[ \vertiii{ \ddot{M}_i((1 - \rho) \theta_0 +\rho \theta_i) - \ddot{M}_i(\theta_0) }^q \right] d\rho \\
\le C(q) L^q \frac{4^q \sqrt{C(2q,d)} G^q}{\lambda^q n^{q/2}} + O(n^{-2}) + O(n^{-4}) = O(n^{-q/2}).
\end{align*}
\end{proof}
Now, recall that we have
\begin{multline}
0=\dot{M}_i(\theta_0) + \ddot{M}_0(\theta_0) [\theta_i - \theta_0] + \left[ \int_{0}^1 \ddot{M}_i((1 - \rho) \theta_0 +\rho \theta_i) d\rho - \ddot{M}_i(\theta_0) \right] [\theta_i - \theta_0] \\
+ [\ddot{M}_i(\theta_0) - \ddot{M}_0(\theta_0)] [\theta_i - \theta_0]. \label{key_eq}
\end{multline}
For the sum of last two terms, we have
\begin{align*}
&\qquad \mathbb{E}\left[ \left\| [\int_{0}^1 \ddot{M}_i((1 - \rho) \theta_0 +\rho \theta_i) d\rho - \ddot{M}_i(\theta_0)] [\theta_i - \theta_0] + [\ddot{M}_i(\theta_0) - \ddot{M}_0(\theta_0)] [\theta_i - \theta_0] \right\|^2 \right] \\
&\le 2\mathbb{E}\left[ \left\| [\int_{0}^1 \ddot{M}_i((1 - \rho) \theta_0 +\rho \theta_i) d\rho - \ddot{M}_i(\theta_0)] [\theta_i - \theta_0] \right\|^2 \right] \\
& \qquad + 2\mathbb{E}\left[ \| [\ddot{M}_i(\theta_0) - \ddot{M}_0(\theta_0)] [\theta_i - \theta_0] \|^2 \right] \qquad \text{(since $(a+b)^2 \le 2a^2 + 2b^2$)}\\
&\le 2 (\mathbb{E} \left[ \left\| \int_{0}^1 \ddot{M}_i((1 - \rho) \theta_0 +\rho \theta_i) d\rho - \ddot{M}_i(\theta_0) \right\|^4 \right])^{1/2} (\mathbb{E}[\| \theta_i - \theta_0 \|^4])^{1/2} \\
& \qquad + 2 (\mathbb{E}[\| \ddot{M}_i(\theta_0) - \ddot{M}_0(\theta_0) \|^4])^{1/2} (\mathbb{E}[\| \theta_i - \theta_0 \|^4])^{1/2} \qquad \text{(H\"{o}lder's inequality)}\\
&= O(n^{-2}) + O(n^{-2}) \qquad \text{(Lemma \ref{matrix_bound} \& \ref{q_bound_of_matrix} and (\ref{q_bound_of_theta_i}))}\\
&= O(n^{-2}).
\end{align*}

\noindent
Unitl now, we have established the upper bound for the mean squared error of local M-estimators,
$$
\mathbb{E} [\| \theta_i - \theta_0 \|^2] \le \frac{2}{n} \mbox{Tr}(\Sigma) + O(n^{-2}),
$$
for $i = 1,\ldots,k$.

\subsection{Bound the Error of Simple Averaging Estimator $\theta^{(0)}$}
Next, we will study the mean squared error of simple averaging estimator,
$$
\theta^{(0)} = \frac{1}{k} \sum_{i=1}^k \theta_i.
$$
We start with a lemma, which bounds the bias of local M-estimator $\theta_i, i=1,\ldots,k$.
\begin{lemma}
\label{bias_bound}
There exists some constant $\tilde{C}>0$ such that for $i=1,\ldots,k$, we have
$$
\| \mathbb{E}[\theta_i - \theta_0] \| \le \frac{\tilde{C}}{n} + O(n^{-2}),
$$
where $\tilde{C} = 16 [C_v(4,d)]^{\frac{1}{4}}\sqrt{C_v(2,d)} \lambda^{-3} G^2L + 4 \sqrt{C_m(2,d)}\sqrt{C_v(2,d)} \lambda^{-2} GH$.
\end{lemma}
\begin{proof}
The main idea of this proof is to use equation (\ref{key_eq}) and apply the established error bounds of Hessian and the aforementioned local m-estimators. By equation (\ref{key_eq}) and fact $\mathbb{E}[M_i(\theta_0)] = 0$, we have
\begin{eqnarray*}
& & \| \mathbb{E}[\theta_i - \theta_0] \| \\
&=& \| \mathbb{E}\{ \ddot{M}_0(\theta_0)^{-1} \left[ \int_{0}^1 \ddot{M}_i((1 - \rho) \theta_0 +\rho \theta_i) d\rho - \ddot{M}_i(\theta_0) \right] [\theta_i - \theta_0] \\
& & \qquad +  \ddot{M}_0(\theta_0)^{-1} [\ddot{M}_i(\theta_0) - \ddot{M}_0(\theta_0)] [\theta_i - \theta_0] \} \| \\
&\le& \| \mathbb{E}\{ \ddot{M}_0(\theta_0)^{-1} \left[ \int_{0}^1 \ddot{M}_i((1 - \rho) \theta_0 +\rho \theta_i) d\rho - \ddot{M}_i(\theta_0) \right] [\theta_i - \theta_0] \} \| \\
& & \qquad + \| \mathbb{E}\{ \ddot{M}_0(\theta_0)^{-1} [\ddot{M}_i(\theta_0) - \ddot{M}_0(\theta_0)] [\theta_i - \theta_0] \} \| \\
&\stackrel{\mbox{Jensen's}}{\le}& \mathbb{E}\left[ \| \ddot{M}_0(\theta_0)^{-1} \left[\int_{0}^1 \ddot{M}_i((1 - \rho) \theta_0 +\rho \theta_i) d\rho - \ddot{M}_i(\theta_0) \right] [\theta_i - \theta_0] \| \right] \\
& & \qquad + \mathbb{E}\left[ \| \ddot{M}_0(\theta_0)^{-1} [\ddot{M}_i(\theta_0) - \ddot{M}_0(\theta_0)] [\theta_i - \theta_0] \| \right] \\
&\stackrel{\mbox{Assumption \ref{invertibility}}}{\le}& \lambda^{-1} \mathbb{E}\left[ \left\| \int_{0}^1 \ddot{M}_i((1 - \rho) \theta_0 +\rho \theta_i) d\rho - \ddot{M}_i(\theta_0) \right\| \| \theta_i - \theta_0 \| \right] \\
& & \qquad + \lambda^{-1} \mathbb{E}\left[ \| \ddot{M}_i(\theta_0) - \ddot{M}_0(\theta_0) \| \| \theta_i - \theta_0 \| \right] \\
&\stackrel{\mbox{H\"{o}lder's}}{\le}&  \lambda^{-1} \mathbb{E} \left[ \left\| \int_{0}^1 \ddot{M}_i((1 - \rho) \theta_0 +\rho \theta_i) d\rho - \ddot{M}_i(\theta_0) \right\|^2 \right]^{1/2} \mathbb{E}[ \| \theta_i - \theta_0 \|^2 ]^{1/2} \\
& & \qquad + \lambda^{-1} \mathbb{E}[ \| \ddot{M}_i(\theta_0) - \ddot{M}_0(\theta_0) \|^2 ]^{1/2} \mathbb{E}[ \| \theta_i - \theta_0 \|^2 ]^{1/2}.
\end{eqnarray*}
Then we can apply Lemma \ref{matrix_bound} \& \ref{q_bound_of_matrix}, and (\ref{q_bound_of_theta_i}) to bound each term, thus, we have
\begin{eqnarray*}
\| \mathbb{E}[\theta_i - \theta_0] \| ]&\le& \lambda^{-1} \sqrt{L^2 \frac{4^2}{\lambda^2}\frac{\sqrt{C_v(4,d)}}{n}G^2 + O(n^{-2})} \sqrt{\frac{4^2}{\lambda^2}\frac{C_v(2,d)}{n}G^2 + O(n^{-4})}\\
& & \qquad + \lambda^{-1} \sqrt{\frac{C_m(2,d)}{n}H^2} \sqrt{\frac{4^2}{\lambda^2}\frac{C_v(2,d)}{n}G^2 + O(n^{-4})} \\
&\le& \lambda^{-1} \left[ L \frac{4}{\lambda} \frac{C_v(4,d)^{1/4}}{\sqrt{n}} G + O(n^{-\frac{3}{2}}) \right] \left[ \frac{4}{\lambda}\frac{\sqrt{C_v(2,d)}}{\sqrt{n}}G + O(n^{-\frac{7}{2}}) \right] \\
& & \qquad + \lambda^{-1} \frac{\sqrt{C_m(2,d)}}{\sqrt{n}}H \left[ \frac{4}{\lambda}\frac{\sqrt{C_v(2,d)}}{\sqrt{n}}G + O(n^{-\frac{7}{2}}) \right] \\
&=& L \frac{4^2}{\lambda^3} \frac{C_v(4,d)^{1/4}\sqrt{C_v(2,d)}}{n} G^2 + O(n^{-2}) \\
& & \qquad + \frac{4}{\lambda^2} \frac{\sqrt{C_m(2,d)}\sqrt{C_v(2,d)}}{n}GH + O(n^{-4}).
\end{eqnarray*}
Let $\tilde{C} = 16 [C_v(4,d)]^{\frac{1}{4}}\sqrt{C_v(2,d)} \lambda^{-3} G^2L + 4 \sqrt{C_m(2,d)}\sqrt{C_v(2,d)} \lambda^{-2} GH$, then we have
$$
\| \mathbb{E}[\theta_i - \theta_0] \| \le \frac{\tilde{C}}{n} + O(n^{-2}).
$$
\end{proof}
\noindent
Then we can show that the MSE of $\theta^{(0)}$ could be bounded as follows.
\begin{lemma}
There exists some constant $\tilde{C}>0$ such that
$$
\mathbb{E} [ \| \theta^{(0)} - \theta_0 \|^2] \le \frac{2}{N} \mbox{Tr}(\Sigma) + \frac{\tilde{C}^2 k^2}{N^2} + O(kN^{-2}) + O(k^3N^{-3}),
$$
where $\tilde{C} = 16 [C_v(4,d)]^{\frac{1}{4}}\sqrt{C_v(2,d)} \lambda^{-3} G^2L + 4 \sqrt{C_m(2,d)}\sqrt{C_v(2,d)} \lambda^{-2} GH$.
\end{lemma}
\begin{proof}
The mean squared error of $\theta^{(0)}$ could be decomposed into two parts: covariance and bias. Thus,
\begin{eqnarray*}
\mathbb{E} [ \| \theta^{(0)} - \theta_0 \|^2] &=& \mbox{Tr}(\mbox{Cov}[\theta^{(0)}]) + \| \mathbb{E}[\theta^{(0)}-\theta_0] \|^2 \\
&=& \frac{1}{k} \mbox{Tr}(\mbox{Cov}[\theta_1]) + \| \mathbb{E}[\theta_1-\theta_0] \|^2 \\
&\le& \frac{1}{k} \mathbb{E} [ \| \theta_1 - \theta_0 \|^2] + \| \mathbb{E}[\theta_1-\theta_0] \|^2,
\end{eqnarray*}
where the first term is well bounded by Lemma \ref{asym_norm2} and the second term could be bounded by Lemma \ref{bias_bound}. Thus, we know
$$
\mathbb{E} [ \| \theta^{(0)} - \theta_0 \|^2] \le \frac{2}{N} \mbox{Tr}(\Sigma) + \frac{\tilde{C}^2 k^2}{N^2} + O(kN^{-2}) + O(k^3N^{-3}).
$$
More generally, for $1 \le q \le 8$, we have
\begin{eqnarray*}
& & \mathbb{E}[ \| \theta^{(0)} - \theta_0 \|^q ] = \mathbb{E}[ \| ( \theta^{(0)} - \mathbb{E}[\theta^{(0)}] ) + ( \mathbb{E}[\theta^{(0)}] - \theta_0 ) \|^q ] \\
&\le& 2^q \mathbb{E}[ \| \theta^{(0)} - \mathbb{E}[\theta^{(0)}] \|^q] + 2^q \| \mathbb{E}[\theta^{(0)}] - \theta_0 \|^q \qquad \text{(since $(a+b)^q \le 2^q a^q + 2^q b^q$)} \\
&=& 2^q \mathbb{E}[ \| \theta^{(0)} - \mathbb{E}[\theta^{(0)}] \|^q ] + 2^q \| \mathbb{E}[\theta_1] - \theta_0 \|^q \qquad \text{(since $\mathbb{E}[\theta^{(0)}] = \mathbb{E}[\theta_1]$)}\\
&\le& 2^q \mathbb{E}[ \| \theta^{(0)} - \theta_0 \|^q ] + 2^q \frac{\tilde{C}^q}{n^q} + O(n^{-q-1}) \\
&\stackrel{\mbox{Lemma \ref{vector_bound}}}{\le}& 2^q \frac{C_v(q,d)}{k^{q/2}} \mathbb{E}[ \| \theta_1 - \theta_0 \|^q ]) + 2^q \frac{\tilde{C}^q}{n^q} + O(n^{-q-1}) \\
&\stackrel{\mbox{(\ref{q_bound_of_theta_i})}}{\le}& 2^q \frac{C_v(q,d)}{k^{q/2}} \left[ \frac{4^q}{\lambda^q}\frac{C(q,d)}{n^{q/2}}G^q + O(n^{-4}) \right] + 2^q \frac{\tilde{C}^q}{n^q} + O(n^{-q-1}) \\
&=& 8^{q} [C_v(q,d)]^2 \lambda^{-q} G^q N^{-\frac{q}{2}} + O(N^{-\frac{q}{2}} n^{\frac{q}{2}-4}) + 2^q \frac{\tilde{C}^q}{n^q} + O(n^{-q-1}).
\end{eqnarray*}
In summary, we have 
\begin{equation} 
\label{q_bound_of_theta(0)}
\mathbb{E}[ \| \theta^{(0)} - \theta_0 \|^q ] \le O(N^{-\frac{q}{2}}) + \frac{2^q \tilde{C}^q k^q}{N^q} + O(k^{q+1}N^{-q-1}) = O(N^{-\frac{q}{2}}) + O(k^{q}N^{-q}) .
\end{equation}
\end{proof}

\section{Proof of Theorem \ref{one_step_M_ap}}
\label{proof_ap}
The whole proof could be completed in two steps: first, show simple averaging estimator $\theta^{(0)}$ is $\sqrt{N}$-consistent when $k=O(\sqrt{N})$; then show the consistency and asymptotic normality of the one-step estimator $\theta^{(1)}$. In the first step, we need to show the following.
\begin{lemma}
\label{N-consistent}
Under Assumption \ref{parameter_space}, \ref{invertibility} and \ref{smoothness}, when $k=O(\sqrt{N})$, the simple averaging estimator $\theta^{(0)}$ is $\sqrt{N}$-consistent estimator of $\theta_0$, i.e.,
$$
\sqrt{N} \| \theta^{(0)} - \theta_0 \| = O_P(1) \text{ as } N \rightarrow \infty.
$$
\end{lemma}
\begin{proof}
If $k$ is finite and does not grow with $N$, the proof is trivial. So, we just need to consider the case that $k \rightarrow \infty$. We know that $\| \mathbb{E}[ \sqrt{n}(\theta_i - \theta_0) ] \| \le O(\frac{1}{\sqrt{n}})$ by Lemma \ref{bias_bound} and $\mathbb{E}[ \| \sqrt{n}(\theta_i - \theta_0) \|^2] \le 2 \mbox{Tr}(\Sigma) + O(n^{-1})$ by Lemma \ref{asym_norm2}. By applying Lindeberg-L\'{e}vy Central Limit Theorem, we have
\begin{multline*}
\sqrt{N} (\theta^{(0)} - \theta_0) = \frac{1}{\sqrt{k}} \sum_{i=1}^k \sqrt{n}(\theta_i - \theta_0) = \frac{1}{\sqrt{k}} \sum_{i=1}^k \{ \sqrt{n}(\theta_i - \theta_0) - \mathbb{E}[\sqrt{n}(\theta_i - \theta_0)] \} + \sqrt{nk}\mathbb{E}[ \theta_1 - \theta_0 ] \\ 
\xrightarrow{d} \bold{N}(0,\Sigma) + \lim_{N \rightarrow \infty} \sqrt{nk}\mathbb{E}[ \theta_1 - \theta_0 ],
\end{multline*}
It suffices to show $\lim_{N \rightarrow \infty} \sqrt{nk}\mathbb{E}[ \theta_1 - \theta_0 ]$ is finite.
By Lemma \ref{bias_bound}, we have
$$
\| \mathbb{E}[\theta_i - \theta_0] \|= O(\frac{1}{n}), \forall i \in \{1,2,\ldots,k\},
$$
which means that $\| \sqrt{nk}\mathbb{E}[ \theta_i - \theta_0 ] \| = O(1)$ if $k=O(\sqrt{N})=O(n)$.
Thus, when $k=O(\sqrt{N})$, $\sqrt{N} (\theta^{(0)} - \theta_0)$ is bounded in probability. 
\end{proof}
Now, we can prove Theorem \ref{one_step_M_ap}.
\begin{proof}
By the definition of the one-step estimator
$$
\theta^{(1)} = \theta^{(0)} - \ddot{M}(\theta^{(0)})^{-1} \dot{M}(\theta^{(0)}),
$$
and by Theorem 4.2 in Chapter XIII of \cite{lang1993real}, we have
\begin{align*}
&\quad \sqrt{N} \ddot{M}(\theta^{(0)}) (\theta^{(1)} - \theta_0) = \ddot{M}(\theta^{(0)}) \sqrt{N} (\theta^{(0)} - \theta_0) - \sqrt{N} (\dot{M}(\theta^{(0)})-\dot{M}(\theta_0)) - \sqrt{N} \dot{M}(\theta_0) \\
&= \ddot{M}(\theta^{(0)}) \sqrt{N} (\theta^{(0)} - \theta_0) - \sqrt{N} \int_{0}^1 \ddot{M}((1 - \rho) \theta_0 +\rho \theta^{(0)}) d\rho \,(\theta^{(0)}-\theta_0) - \sqrt{N} \dot{M}(\theta_0)\\
&= \left[ \ddot{M}(\theta^{(0)}) - \int_{0}^1 \ddot{M}((1 - \rho) \theta_0 +\rho \theta^{(0)}) d\rho \right] \sqrt{N} (\theta^{(0)} - \theta_0) - \sqrt{N} \dot{M}(\theta_0),
\end{align*}
As it is shown in (\ref{q_bound_of_theta(0)}), for any $\rho \in [0,1]$, when $k = O(\sqrt{N})$, we have
$$\| (1 - \rho) \theta_0 +\rho \theta^{(0)} - \theta_0 \| \le \rho \| \theta^{(0)} - \theta_0 \| \xrightarrow{P} 0.$$
Since $\ddot{M}(\cdot)$ is a continuous function, $\vertiii{ \ddot{M}(\theta^{(0)}) - \int_{0}^1 \ddot{M}((1 - \rho) \theta_0 +\rho \theta^{(0)}) d\rho }\xrightarrow{P} 0$. Thus,
$$
\sqrt{N} \ddot{M}(\theta^{(0)}) (\theta^{(1)} - \theta_0) = - \sqrt{N} \dot{M}(\theta_0) + o_P(1).
$$
And, $\ddot{M}(\theta^{(0)}) \xrightarrow{P} \ddot{M}_0(\theta_0)$ because of $\theta^{(0)} \xrightarrow{P} \theta_0$ and Law of Large Number. Therefore, we can obatain
$$
\sqrt{N} (\theta^{(1)} - \theta_0) \xrightarrow{d} \bold{N}(0,\Sigma) \text{ as } N \rightarrow \infty
$$
by applying Slutsky's Lemma.
\end{proof}

\section{Proof of Theorem \ref{one_step_M}}
\label{proof}
Let us recall the formula for one-step estimator,
$$
\theta^{(1)} = \theta^{(0)} - \ddot{M}(\theta^{(0)})^{-1} \dot{M}(\theta^{(0)}).
$$
Then by Theorem 4.2 in Chapter XIII of \cite{lang1993real}, we have
\begin{eqnarray*}
& & \ddot{M}_0(\theta_0) (\theta^{(1)} - \theta_0) = [\ddot{M}_0(\theta_0)-\ddot{M}(\theta^{(0)})] (\theta^{(1)} - \theta_0) + \ddot{M}(\theta^{(0)}) (\theta^{(1)} - \theta_0) \\
&=& [\ddot{M}_0(\theta_0)-\ddot{M}(\theta^{(0)})] (\theta^{(1)} - \theta_0) + \ddot{M}(\theta^{(0)}) (\theta^{(0)} - \theta_0) - [\dot{M}(\theta^{(0)}) - \dot{M}(\theta_0)] - \dot{M}(\theta_0) \\
&=& [\ddot{M}_0(\theta_0)-\ddot{M}(\theta^{(0)})] (\theta^{(1)} - \theta_0) + \left[ \ddot{M}(\theta_0) - \int_{0}^1 \ddot{M}((1 - \rho) \theta_0 +\rho \theta^{(0)}) d\rho \right] (\theta^{(0)} - \theta_0) - \dot{M}(\theta_0).
\end{eqnarray*}
Then we have
\begin{multline}
\label{key_eq1}
\theta^{(1)} - \theta_0 =   - \ddot{M}_0(\theta_0)^{-1} \dot{M}(\theta_0) + \ddot{M}_0(\theta_0)^{-1} [\ddot{M}_0(\theta_0)-\ddot{M}(\theta^{(0)})] (\theta^{(1)} - \theta_0) \\
 + \ddot{M}_0(\theta_0)^{-1} \left[ \ddot{M}(\theta_0) - \int_{0}^1 \ddot{M}((1 - \rho) \theta_0 +\rho \theta^{(0)}) d\rho \right] (\theta^{(0)} - \theta_0)
\end{multline}
We will show the last two terms are small enough. Similar to the proof of Lemma \ref{asym_norm2}, we define a ``good" event:
\begin{eqnarray*}
E_4 &=& \{ \| \theta^{(0)} - \theta_0 \| \le \delta \}.
%E_5 &=& \{ \frac{1}{N} \sum_{i=1}^k \sum_{x \in S_i} L(x) \le 2L \}, \\
%E_6 &=& \{ \vertiii{ \ddot{M}(\theta^{(0)}) } \ge \lambda/4 \}
\end{eqnarray*}
The probability of above event is close to 1 when $N$ is large.
\begin{eqnarray*}
\mbox{Pr}(E_4^c) \le \frac{\mathbb{E}[\| \theta^{(0)} - \theta_0 \|^8]}{\delta^8} \le O(N^{-4}) + O(k^8 N^{-8}).
\end{eqnarray*}

\begin{lemma}
If event $E_4$ holds, for $1 \le q \le 4$, we have
\begin{eqnarray*}
\mathbb{E} \left[ \vertiii{\ddot{M}_0(\theta_0)-\ddot{M}(\theta^{(0)})}^q \right] &\le& O(N^{-\frac{q}{2}}) + O(k^{q}N^{-q}), \\
\mathbb{E} \left[ \vertiii{\ddot{M}(\theta_0) - \int_{0}^1 \ddot{M}((1 - \rho) \theta_0 +\rho \theta^{(0)}) d\rho}^q \right] &\le& O(N^{-\frac{q}{2}}) + O(k^{q}N^{-q}).
\end{eqnarray*}
\end{lemma}
\begin{proof}
By Lemma \ref{matrix_bound}, we know
$$
\mathbb{E} \left[ \vertiii{ \ddot{M}_0(\theta_0)- \ddot{M}(\theta_0) }^q \right] \le \frac{C(q,d)}{N^{q/2}} H^q.
$$
Under event $E_4$ and Assumption \ref{smoothness}, by applying Jensen's inequality, we have
$$
\vertiii{ \ddot{M}(\theta_0) - \ddot{M}(\theta^{(0)}) }^q \le \frac{1}{N} \sum_{i=1}^k \sum_{x \in S_i} \vertiii{ \ddot{m}(x;\theta^{(0)}) - \ddot{m}(x;\theta_0) }^q \le \frac{1}{N} \sum_{i=1}^k \sum_{x \in S_i} L(x)^q \| \theta^{(0)} - \theta_0 \|^q.
$$
Thus, for $1 \le q \le 4$, we have
\begin{eqnarray*}
\mathbb{E} \left[ \vertiii{ \ddot{M}(\theta_0) - \ddot{M}(\theta^{(0)}) }^q \right] &\stackrel{\mbox{H\"{o}lder's}}{\le}& \mathbb{E} \left[ (\frac{1}{N} \sum_{i=1}^k \sum_{x \in S_i} L(x)^q)^2 \right]^{\frac{1}{2}} \mathbb{E} \left[ \| \theta^{(0)} - \theta_0 \|^{2q} \right]^{\frac{1}{2}} \\
&\stackrel{\mbox{(\ref{q_bound_of_theta(0)})}}{\le}& O(N^{-\frac{q}{2}}) + \frac{2^q \tilde{C}^q L^q k^q}{N^q} + O(\frac{k^{q+1}}{N^{q+1}}).
\end{eqnarray*}
As a result, we have, for $1 \le q \le 4$,
\begin{eqnarray*}
\mathbb{E} \left[ \vertiii{\ddot{M}_0(\theta_0)-\ddot{M}(\theta^{(0)})}^q \right] &\le& 2^q \mathbb{E} \left[ \vertiii{\ddot{M}_0(\theta_0)-\ddot{M}(\theta_0))}^q \right] + 2^q \mathbb{E} \left[ \vertiii{ \ddot{M}(\theta_0) - \ddot{M}(\theta^{(0)}) }^q \right] \\
&\le& O(N^{-\frac{q}{2}}) + \frac{4^q \tilde{C}^q L^q k^q}{N^q} + O(\frac{k^{q+1}}{N^{q+1}}).
\end{eqnarray*}
In this proof, we let $\theta' = (1 - \rho) \theta_0 +\rho \theta^{(0)}$ for the simplicity of notation.
Note that $\theta'- \theta_0 = \rho (\theta^{(0)} - \theta_0)$, then by event $E_4$, Assumption \ref{smoothness} and inequality (\ref{q_bound_of_theta(0)}), we have
\begin{eqnarray*}
\mathbb{E} \left[ \vertiii{ \ddot{M}(\theta_0) - \ddot{M}(\theta') }^q \right] &\stackrel{\mbox{Jensen's}}{\le}& \mathbb{E} \left[ \frac{1}{N} \sum_{i=1}^k \sum_{x \in S_i} L(x)^q \| \theta' - \theta_0 \|^q \right] \\
&\stackrel{\mbox{H\"{o}lder's}}{\le}& \mathbb{E} \left[ (\frac{1}{N} \sum_{i=1}^k \sum_{x \in S_i} L(x)^q)^2 \right]^{\frac{1}{2}} \rho^q \mathbb{E} \left[ \| \theta^{(0)} - \theta_0 \|^{2q} \right]^{\frac{1}{2}} \\
&\le& O(N^{-\frac{q}{2}}) + \frac{2^q \tilde{C}^q L^q k^q}{N^q} + O(\frac{k^{q+1}}{N^{q+1}}).
\end{eqnarray*}
So, we have
\begin{align*}
& \mathbb{E} \left[ \vertiii{\ddot{M}(\theta_0) - \int_{0}^1 \ddot{M}((1 - \rho) \theta_0 +\rho \theta^{(0)}) d\rho}^q \right] \le \mathbb{E}\left[ \int_{0}^1 \vertiii{ \ddot{M}(\theta_0) - \ddot{M}((1 - \rho) \theta_0 +\rho \theta^{(0)}) }^q d\rho \right] \\
=& \int_{0}^1 \mathbb{E}\left[ \vertiii{ \ddot{M}(\theta_0) - \ddot{M}((1 - \rho) \theta_0 +\rho \theta^{(0)}) }^q \right] d\rho \le O(N^{-\frac{q}{2}}) + \frac{2^q \tilde{C}^q L^q k^q}{N^q} + O(\frac{k^{q+1}}{N^{q+1}}).
\end{align*}
\end{proof}
Therefore, under event $E_4$, for $1 \le q \le 4$, we can bound $\ddot{M}_0(\theta_0)^{-1} [\ddot{M}_0(\theta_0)-\ddot{M}(\theta^{(0)})] (\theta^{(1)} - \theta_0)$ and $\ddot{M}_0(\theta_0)^{-1} [\ddot{M}(\theta_0) - \int_{0}^1 \ddot{M}((1 - \rho) \theta_0 +\rho \theta^{(0)}) d\rho] (\theta^{(0)} - \theta_0)$ as follows:
\begin{eqnarray*}
\mathbb{E}[ \| \ddot{M}_0(\theta_0)^{-1} [\ddot{M}_0(\theta_0)-\ddot{M}(\theta^{(0)})] (\theta^{(1)} - \theta_0) \|^q ] 
&\le& \lambda^{-q} \mathbb{E} \left[ \vertiii{\ddot{M}_0(\theta_0)-\ddot{M}(\theta^{(0)})}^q \right] D^q \\
&\le& O(N^{-\frac{q}{2}}) + O(k^{q}N^{-q}).
\end{eqnarray*}
and,
\begin{eqnarray*}
& & \mathbb{E} \left[ \| \ddot{M}_0(\theta_0)^{-1} [\ddot{M}(\theta_0) - \int_{0}^1 \ddot{M}((1 - \rho) \theta_0 +\rho \theta^{(0)}) d\rho] (\theta^{(0)} - \theta_0) \|^q \right] \\
&\le& \lambda^{-q}\mathbb{E} \left[ \vertiii{\ddot{M}(\theta_0) - \int_{0}^1 \ddot{M}((1 - \rho) \theta_0 +\rho \theta^{(0)}) d\rho}^q \right] D^q \\
&\le& O(N^{-\frac{q}{2}}) + O(k^{q}N^{-q}).
\end{eqnarray*}
And by Lemma \ref{vector_bound}, for $1 \le q \le 8$, we have
$$
\mathbb{E} [ \| \dot{M}(\theta_0) \|^q ] = O(N^{-\frac{q}{2}}).
$$
Therefore, combining above three bounds and equation (\ref{key_eq1}), we have, for $1 \le q \le 4$,
\begin{eqnarray*}
\mathbb{E}[ \| \theta^{(1)} - \theta_0 \|^q ] &\le& 3^q \mathbb{E}[ \| \ddot{M}_0(\theta_0)^{-1} [\ddot{M}_0(\theta_0)-\ddot{M}(\theta^{(0)})] (\theta^{(1)} - \theta_0) \|^q ] \\
& & + 3^q \mathbb{E}[ | \ddot{M}_0(\theta_0)^{-1} \left[ \ddot{M}(\theta_0) - \int_{0}^1 \ddot{M}((1 - \rho) \theta_0 +\rho \theta^{(0)}) d\rho \right] (\theta^{(0)} - \theta_0) \|^q ] \\
& & + 3^q \mathbb{E} [ \| \ddot{M}_0(\theta_0)^{-1} \dot{M}(\theta_0) \|^q ] + Pr(E_4^c) D^q\\
&=& O(N^{-\frac{q}{2}}) + O(k^{q}N^{-q}).
\end{eqnarray*}
Now, we can give tighter bounds for the first two terms in equation (\ref{key_eq1}) by H\"{o}lder's inequality.
\begin{eqnarray*}
& & \mathbb{E}[ \| \ddot{M}_0(\theta_0)^{-1} [\ddot{M}_0(\theta_0)-\ddot{M}(\theta^{(0)})] (\theta^{(1)} - \theta_0) \|^2 ] \\
&\le& \lambda^{-2} \sqrt{ \mathbb{E} \left[ \vertiii{\ddot{M}_0(\theta_0)-\ddot{M}(\theta^{(0)})}^4 \right] } \sqrt{ \mathbb{E}[ \| \theta^{(1)} - \theta_0 \|^4] } \\
&=& O(N^{-2}) + O(k^{4}N^{-4}),
\end{eqnarray*}
and,
\begin{eqnarray*}
& & \mathbb{E}[ \| \ddot{M}_0(\theta_0)^{-1} [\ddot{M}(\theta_0) - \int_{0}^1 \ddot{M}((1 - \rho) \theta_0 +\rho \theta^{(0)}) d\rho] (\theta^{(0)} - \theta_0) \|^2 ] \\
&\le& \lambda^{-2} \sqrt{ \mathbb{E} \left[ \vertiii{\ddot{M}(\theta_0)-\int_{0}^1 \ddot{M}((1 - \rho) \theta_0 +\rho \theta^{(0)}) d\rho}^4 \right] } \sqrt{ \mathbb{E}[ \| \theta^{(0)} - \theta_0 \|^4] } \\
&=& O(N^{-2}) + O(k^{4}N^{-4}).
\end{eqnarray*}
Now, we can finalize our proof by using equation (\ref{key_eq}) again,
\begin{eqnarray*}
\mathbb{E}[ \| \theta^{(1)} - \theta_0 \|^2 ] &\le& 2 \mathbb{E}[ \| \ddot{M}_0(\theta_0)^{-1} \dot{M}(\theta_0) \|^2] + 4 \mathbb{E}[ \| \ddot{M}_0(\theta_0)^{-1} [\ddot{M}_0(\theta_0)-\ddot{M}(\theta^{(0)})] (\theta^{(1)} - \theta_0) \|^2 ] \\
& & + 4 \mathbb{E} \left[ \| \ddot{M}_0(\theta_0)^{-1} [\ddot{M}(\theta_0) - \int_{0}^1 \ddot{M}((1 - \rho) \theta_0 +\rho \theta^{(0)}) d\rho] (\theta^{(0)} - \theta_0) \|^2 \right] \\
&\le& 2 \frac{\mbox{Tr}(\Sigma)}{N} + O(N^{-2}) + O(k^{4}N^{-4}).
\end{eqnarray*}

\begin{comment}
\begin{remark*}
\cite{zhang2013communication} also propose an additonal resampling step to reduce the bias of the simple averaging estimator and call the corresponding estimator, Subsampled Mixture Averaging Estimator (SAVGM). The mean squared error of SAVGM is bounded by $O(N^{-1})+O(kN^{-2})+O(k^3N^{-3})$ while the error of one-step estimator could be bounded in lower order, $O(N^{-1})+O(N^{-2})+O(k^4N^{-4})$. And it is worth noting that SAVGM cannot achieve asymptotic efficiency because the resampling step inflates the estimator covariance although the bias is reduced.
\end{remark*}
\end{comment}

\begin{comment}
By multiplying $\sqrt{N}$ on both sides of equation (\ref{key_eq1}), we obtain
\begin{equation*}
\begin{split}
\sqrt{N}(\theta^{(1)} - \theta_0) =  &\; \sqrt{N} \ddot{M}_0(\theta_0)^{-1} [\ddot{M}_0(\theta_0)-\ddot{M}(\theta^{(0)})] (\theta^{(1)} - \theta_0) \\
& + \sqrt{N} \ddot{M}_0(\theta_0)^{-1} [\ddot{M}(\theta_0) - \ddot{M}(\theta'')] (\theta^{(0)} - \theta_0) - \ddot{M}_0(\theta_0)^{-1} \sqrt{N} \dot{M}(\theta_0).
\end{split}
\end{equation*}
Note that the first two terms on the RHS converge to $0$ in probability when $k = O(\sqrt{N})$. We also have 
$$\sqrt{N} \dot{M}(\theta_0) \xrightarrow{d} N(0,\mathbb{E}[\dot{m}(x;\theta_0)\dot{m}(x;\theta_0)^t]) $$ 
by Central Limit Theorem (CLT). Therefore,
$$
\sqrt{N}(\theta^{(1)} - \theta_0) = -\ddot{M}_0(\theta_0)^{-1} \sqrt{N} \dot{M}(\theta_0) + o_P(1) \xrightarrow{d} N(0,\Sigma).
$$
\end{comment}

\section{Proof of Corollary \ref{estimation_with_crash}}
\label{proof_crash}
At first, we will present a lemma on the negative moments of a Binomial random variable, i.e., $\mathbb{E} \left[ \frac{1}{Z}1_{(Z>0)} \right]$ and $\mathbb{E} \left[ \frac{1}{Z^2} 1_{(Z>0)} \right]$, where $Z \sim \bold{B}(k,p)$ and $\bold{B}(k,p)$ denotes Binomial distribution with $k$ independent trials and a success probability $p$ for each trial. We believe that $\mathbb{E} \left[ \frac{1}{Z}1_{(Z>0)} \right]$ and $\mathbb{E}\left[ \frac{1}{Z^2} 1_{(Z>0)} \right]$ should have been well studied. However, we did not find any appropriate reference on their upper bounds that we need.
So, we derive the upper bounds as follows, which will be useful in the proof of Corollary \ref{estimation_with_crash}.
\begin{lemma}
\label{bound_binomial}
Suppose $Z \sim \bold{B}(k,p)$, when $z>0$, we have
$$
\mathbb{E} \left[ \frac{1}{Z} 1_{(Z>0)} \right] < \frac{1}{kp} + \frac{3}{k^2 p^2} \text{\quad and \quad}
\mathbb{E} \left[ \frac{1}{Z^2} 1_{(Z>0)} \right] < \frac{6}{k^2 p^2}.
$$
\end{lemma}
\begin{proof}
By definition, we have
\begin{eqnarray*}
& & \mathbb{E} \left[ \frac{1}{Z} 1_{(Z>0)} \right] = \sum_{z=1}^k \frac{1}{z} \binom{k}{z} p^z (1-p)^{k-z} = \sum_{z=1}^k \frac{1}{z} \frac{k!}{z!(k-z)!} p^z (1-p)^{k-z} \\
&=& \sum_{z=1}^k \frac{z+1}{z} \frac{1}{(k+1)p}\frac{(k+1)!}{(z+1)!(k-z)!} p^{z+1} (1-p)^{k-z} \\
&=& \sum_{z=1}^k \frac{1}{(k+1)p} \binom{k+1}{z+1} p^{z+1} (1-p)^{k-z} + \sum_{z=1}^k \frac{1}{z} \frac{1}{(k+1)p} \binom{k+1}{z+1} p^{z+1} (1-p)^{k-z} \\
&<& \frac{1}{(k+1)p} + \sum_{z=1}^k \frac{z+2}{z} \frac{1}{(k+1)(k+2)p^2} \binom{k+2}{z+2} p^{z+2} (1-p)^{k-z} \\
&<& \frac{1}{(k+1)p} + \frac{3}{(k+1)(k+2)p^2} < \frac{1}{kp} + \frac{3}{k^2 p^2}.
\end{eqnarray*}
Similarly, we have
\begin{eqnarray*}
& & \mathbb{E} \left[ \frac{1}{Z^2} 1_{(Z>0)} \right] = \sum_{z=1}^k \frac{1}{z^2} \binom{k}{z} p^z (1-p)^{k-z} = \sum_{z=1}^k \frac{1}{z^2} \frac{k!}{z!(k-z)!} p^z (1-p)^{k-z} \\
&=& \sum_{z=1}^k \frac{(z+1)(z+2)}{z^2} \frac{1}{(k+1)(k+2)p^2} \frac{(k+2)!}{(z+2)!(k-z)!} p^{z+2} (1-p)^{k-z} \\
&\le& \sum_{z=1}^k \frac{6}{(k+1)(k+2)p^2} \binom{k+2}{z+2} p^{z+2} (1-p)^{k-z} \\
&<& \frac{6}{(k+1)(k+2)p^2} < \frac{6}{k^2 p^2}.
\end{eqnarray*}
\end{proof}

Now, we can prove Corollary \ref{estimation_with_crash} could be as follows.
\begin{proof}
Let the random variable $Z$ denote the number of machines that successfully communicate with the central machine, which means that $Z$ follows Binomial distribution, $B(k,1-r)$. By Law of Large Number, $\frac{Z}{(1-r)k} \xrightarrow{P} 1\text{ as } k \rightarrow \infty$. If $Z$ is known, the size of available data becomes $Zn$. By Theorem \ref{one_step_M_ap}, the one-step estimator $\theta^{(1)}$ is still asymptotic normal when $k=O(\sqrt{N})$,
$$
\sqrt{Zn}(\theta^{(1)} - \theta_0) \xrightarrow{d} N(0,\Sigma) \text{ as } n \rightarrow \infty.
$$
Therefore, when $k \rightarrow \infty$, we have
$$
\sqrt{(1-r)N}(\theta^{(1)} - \theta_0) = \sqrt{\frac{(1-r)N}{Zn}} \sqrt{Zn}(\theta^{(1)} - \theta_0) \xrightarrow{d} \sqrt{\frac{(1-r)N}{Zn}} N(0,\Sigma).
$$
Since $\frac{(1-r)N}{Zn}=\frac{(1-r)k}{Z} \xrightarrow{P} 1$, by Slutsky's Lemma, we have
$$
\sqrt{(1-r)N}(\theta^{(0)} - \theta_0) \xrightarrow{d} N(0,\Sigma).
$$
This result indicates that when the local machines could lose communication independently with central machine with probability $q$, the one-step estimator $\theta^{(1)}$ shares the same asymptotic properties with the oracle M-estimator using $(1-r) \times 100\%$ of the total samples.

Next, we will analyze the mean squared error of one-step estimator with the presence of local machine failures. Note that, when $Z$ is fixed and known, by Theorem \ref{one_step_M}, we have
$$
\mathbb{E}[ \| \theta^{(1)} - \theta_0 \|^2 \big{|} Z] \le \frac{2\mbox{Tr}[\Sigma]}{nZ} + O(n^{-2}Z^{-2}) + O(n^{-4}).
$$
By Rule of Double Expectation and Lemma \ref{bound_binomial},
\begin{eqnarray*}
& & \mathbb{E}[ \| \theta^{(1)} - \theta_0 \|^2 1_{(Z>0)}] = \mathbb{E}[ \mathbb{E}[ \| \theta^{(1)} - \theta_0 \|^2 \big{|} Z] 1_{(Z>0)} ] \\
&\le& \mathbb{E} \left[ \frac{2\mbox{Tr}[\Sigma]}{nZ} 1_{(Z>0)} \right] +\mathbb{E}[ (O(n^{-2}Z^{-2}) + O(n^{-4}) ) 1_{(Z>0)} ] \\
&\le& 2\mbox{Tr}[\Sigma] \left\{ \frac{1}{nk(1-r)} + \frac{3}{nk^2 (1-r)^2} \right\} + O(n^{-2} k^{-2}(1-r)^{-2}) + O(n^{-4})\\
&=& \frac{2\mbox{Tr}[\Sigma]}{N(1-r)} + \frac{6\mbox{Tr}[\Sigma]}{N k (1-r)^2} + O(N^{-2}(1-r)^{-2}) + O(k^2 N^{-2}).
\end{eqnarray*}
\end{proof}

\end{document}